\RequirePackage{fix-cm}
\documentclass[smallextended]{svjour3}       
\smartqed  
\usepackage{amsmath, amsfonts}
\usepackage{amssymb}
\usepackage{bm}
\usepackage{url}
\usepackage{eurosym}
\usepackage{graphicx}
\usepackage{xcolor}
\DeclareMathOperator*{\argmax}{arg\,max}

\newcommand{\Po}{\mathbb{P}}
\newcommand{\F}{\mathcal{F}}
\newcommand{\R}{\mathbb{R}}

\spnewtheorem{rem}{Remark}{\bf}{\it}
\spnewtheorem{prob}{Problem}{\bf}{\it}
\spnewtheorem{thm}{Theorem}{\bf}{\it}
\spnewtheorem{prop}{Proposition}{\bf}{\it}
\spnewtheorem{asu}{Assumption}{\bf}{\it}
\journalname{Mathematics and Financial Economics}
\begin{document}

\title{Optimal investment with transaction costs\\ under cumulative prospect theory in discrete time
}


\titlerunning{Optimal Investment with Transaction Costs under CPT}        

\author{Bin Zou         \and
        Rudi Zagst 
}


\institute{Bin Zou \at
              Department of Applied Mathematics, University of Washington\\
              Tel.: +1 206-543-4065\\
              Fax: +1 206-685-1440\\
              \email{binzou@uw.edu}
           \and
           Rudi Zagst (Corresponding author) \at
           Chair of Mathematical Finance, Technical University of Munich\\
           Parkring 11, Garching-Hochbr\"uck 85748, Germany\\
              Tel.: +49 89-289-17401\\
              Fax: +49 89-289-17407\\
              \email{zagst@tum.de}
}

\date{Received: April 27, 2016\\
Revised: November 13, 2016}

\maketitle

\begin{abstract}
We study optimal investment problems under the framework of cumulative prospective theory (CPT).
A CPT investor makes investment decisions in a single-period financial market with transaction costs.
The objective is to seek the optimal investment strategy that maximizes the prospect value of the investor's final wealth.
We obtain the optimal investment strategy explicitly in two examples.
An economic analysis is conducted to investigate the impact of the transaction costs and risk aversion on the optimal investment strategy.
\keywords{cumulative prospect theory \and economic analysis \and optimal investment \and S-shaped utility \and transaction costs}

\vspace{1em}
\noindent
\textbf{JEL Classification} G11
\end{abstract}

\section{Introduction}
\label{intro}
In economics and finance, an essential problem is how to model people's preference over uncertain outcomes.
To address this problem, \cite{B54} (originally published in 1738) proposes \emph{expected utility theory} (EUT): any uncertain outcome $X$ is represented by a numerical value $E[U(X)]$, which is the expected value of the utility $U(X)$ taken under an objective probability measure $\Po$.
An outcome $X_1$ is preferred to another outcome $X_2$ if and only if
$E[U(X_1)] > E[U(X_2)]$. Hence, according to EUT, a rational individual seeks to maximize the expected utility $E[U(X)]$ over all available outcomes.
Bernoulli's original EUT is formally established by von Neumann and Morgenstern (thus the theory is also called \emph{von Neumann-Morgenstern utility theorem})
in \cite{VNM44}, which show that any individual whose behavior satisfies certain axioms has a utility function $U$ and always prefers outcomes that maximize the expected utility.
Since then, expected utility maximization has been one of the most widely used criteria for optimization problems concerning uncertainty,
see, e.g., the pioneering work of \cite{M69} and \cite{S69} on optimal investment problems.

However, empirical experiments and research show that human behavior may violate the basic tenets of EUT, e.g., Allais paradox challenges the fundamental \emph{independence axiom} of EUT.
In addition, the utility function $u(\cdot)$ under EUT is concave, i.e. individuals are uniformly risk averse, which contradicts the risk seeking behavior in case of losses observed from behavioral experiments.
Please refer to \cite{KT79} for many designed choice problems and results which cannot be explained by EUT.
Alternative theories have been proposed to address the drawbacks of EUT, such as \emph{Prospect Theory} by \cite{KT79}, \emph{Rank-Dependent Utility} by \cite{Q82}, and \emph{Cumulative Prospect Theory}\footnote{Daniel Kahneman was awarded the 2002 Nobel Memorial Prize in Economic Sciences for his pioneering work on the psychology of decision-making and behavioral economics (notably, prospect theory and cumulative prospect theory).} (CPT) by \cite{TK92}.
CPT can explain diminishing sensitivity, loss aversion, and different risk attitudes.
Furthermore, unlike prospect theory, CPT does not violate the first-order
stochastic dominance.
The detailed characterizations of CPT are presented in Subsection \ref{subsection_CPT}.
In the CPT setting, the objective functional is non-concave and non-convex;
in addition, the probability distortion (an intrinsic feature of CPT) destroys the tower rule of conditional expectation.
Hence, two powerful tools, martingale method (convex duality) and dynamic programming, commonly used to solve optimization problems under EUT are not longer applicable under CPT.

Optimal investment problems under the CPT framework (in both continuous time models and discrete time models) have attracted attentions recently in the academic field, although the related literature is still scarce comparing to the vast extent of literature on optimal investment problems under EUT.
The first attempt to solve optimal investment problems under prospect theory (excluding probability distortion in CPT) in continuous time can be traced back to \cite{BKP04}, in which the optimal portfolio weight is found explicitly in a complete market without transaction costs.
Jin and Zhou in \cite{JZ08} provide the first analytical treatment for the same problem as in \cite{BKP04}, but under a full CPT setting (namely, with the probability distortion feature), and also address
the well postedness of the problem.
By splitting the CPT optimization problem into two Choquet optimization problems and applying a quantile transformation, the authors of \cite{JZ08} obtain an explicit characterization for the terminal value of the optimal portfolio and thus the existence (and also uniqueness) of the optimal investment strategy in a complete continuous model.
Using the tool of quantile transformation to solve portfolio selection problems is fully studied in \cite{CD11} and \cite{HZ11b} under a wide collection of optimization criteria, including CPT.
The sufficient and necessary conditions for a well posted optimal investment problem under CPT and the existence of an optimal investment strategy are further studied in \cite{RR13} for a piece-wise power probability distortion and a piece-wise power utility function, and in \cite{RR14} for bounded utility on gains (for example, an exponential utility).
An asymptotic analysis on optimal investment strategies is performed in \cite{JZ13}.
Nonetheless, in all above mentioned papers, neither an analytical expression nor a numerical method is provided to solve for the optimal investment strategy in a continuous market.
Furthermore, an essential assumption for all those papers in continuous time models is the market completeness,
or equivalently the uniqueness of the pricing kernel.

Studying CPT portfolio optimization problems in discrete time is as hard as analyzing those in continuous time, and even requires different techniques and tools, since discrete time models are intrinsically incomplete.
The initial work on such problems in a single-period discrete model (again without transaction costs) is done in two parallel papers, \cite{BG10} and \cite{HZ11}.
The authors of \cite{BG10} obtain an explicit optimal solution in a frictionless financial market under the following assumptions: a piece-wise power utility, risk-free asset as the reference point, and the constraint of no short-selling.
They also study the properties of a new risk measure (called \emph{CPT-ratio}) and conduct numerical simulations to investigate the impact of several factors, including mean, volatility, skewness, and risk aversion, on the optimal investment.
In \cite{HZ11}, the authors consider the same problem as in \cite{BG10}, but provide detailed analysis on the well posedness of the problem by introducing a new measure of loss aversion (called \emph{large-loss aversion degree}). They do not impose any constraint on the investment strategies and are able to find optimal solutions explicitly in two cases:
(1) a piece-wise power utility and risk-free asset as the reference point;
(2) a piece-wise linear utility and a general reference point.
In \cite{HZ14}, upon the model of \cite{HZ11}, the authors use the NYSE equity and US treasury bond returns for the period 1926-1990 to perform empirical studies, including the impact of loss aversion, evaluation period, and the reference point on the optimal investment strategy.
In both \cite{BG10} and \cite{HZ11}, there is only one risky asset in the financial market.
The authors of \cite{PS12} then extend the previous work by considering a frictionless market consisting of one risk-free asset and multiple risky assets.
Their main contribution is to provide a two-fund separation theorem between the risk-free asset and the market portfolio when the excess return has an elliptically symmetric distribution.
In \cite{CR15}, the authors tackle the CPT portfolio optimization problem in a multi-period discrete market model for the first time in literature.
They not only address the well posedness issue of the problem but also establish the existence of optimal strategies under some assumptions.
A similar problem as in \cite{CR15} without probability distortion is also considered in \cite{CR15b}, where dynamic programming is applied to solve the maximization problems of non-concave utility functions over the entire real line. In \cite{CRR15}, the authors investigate non-concave utility maximization problems on the positive real line, i.e., they restrict to portfolio strategies that lead to non-negative terminal wealth.
Along the same line as \cite{CR15b} and \cite{CRR15}, \cite{R15} studies the same problem but for non-concave utility that is bounded from above.

Without transaction costs, the optimal portfolio found in Merton's framework may lead to unrealistic strategies, e.g., buying stocks at infinite amount. In real life, transaction costs (bid-ask spread) are always present, albeit small for highly liquid assets.
In their seminal paper, \cite{MC76} claim through heuristic arguments that the optimal portfolio contains a no-trade region. \cite{DN90} provide a rigorous treatment on optimal policies, solutions to the free boundary problem, and the value process (associated optimal expected utility).
\cite{SS94} further generalize the results with viscosity techniques.
Please refer to \cite{KS10} and the references therein for a comprehensive introduction and development on the mathematical theory of financial markets with transaction costs.
The majority of existing literature on optimal investment problems with transaction \mbox{costs}, including those mentioned above, pursue an analysis for an investor who \mbox{behaves} according to EUT.

In this paper, we consider optimal investment problems under the CPT framework in a single-period discrete-time model with transaction costs, which, to our best knowledge, has not been studied before.
Our work differs from the existing literature in several directions.
We summarize the main contributions of our paper through detailed comparisons to the literature in what follows.
\begin{itemize}
  \item First, the financial market we consider is not only incomplete but also non frictionless.
As reviewed in the above context, current literature in continuous time relies heavily on the market completeness while in discrete time all existing papers, as far as we know, work under a frictionless market (i.e., a market without transaction costs).
For the first time in the literature, our paper provides studies on CPT portfolio optimization problems in a financial market with transaction costs.

  \item Secondly, our main objective is to obtain optimal investment strategies in explicit form.
We are able to achieve such objective in two market cases.
Several existing papers, e.g., \cite{BG10}, \cite{HZ11}, and \cite{PS12}, pursue a similar goal as ours, but under different settings of the utility function, the probability distortion, the reference point, and/or the distribution of the risky return.
For instance, in Section \ref{section_case1}, we consider a random reference point while \cite{BG10} and \cite[Section 5.1]{HZ11} assume a constant reference point;
in Section \ref{sec_case2}, we consider a piece-wise exponential utility function (which is bounded from above) but \cite{BG10} and \cite[Section 5.1]{HZ11} work with a piece-wise power utility (unbounded from above);
in Section \ref{sec_case2}, the risky return follows a binomial (discrete) distribution, whereas an elliptically symmetric (continuous) distribution is considered in \cite{PS12}.

\item Thirdly, we perform an economic analysis for sensitivity behavior of the optimal investment strategy.
Such studies are not available in continuous time models, since finding an explicit optimal portfolio is still an open question, even in the simplest Black-Scholes model.
In discrete time models, sensitivity results exist in a very limited extent of literature, for instance, \cite[Section 5]{BG10} and \cite[Section 5]{PS12}.
Due to the different market settings as \cite{BG10} and \cite{PS12}, the sensitivity analysis performed in our paper complements to the current literature.
\end{itemize}

The rest of the paper is organized as follows.
Section \ref{sec_model} introduces the market model with transaction costs, and the three key components of the CPT framework.
The main optimization problem is also formulated in Section \ref{sec_model}.
We obtain the optimal investment strategy in explicit form for two cases in Section \ref{section_case1} and Section \ref{sec_case2}, respectively.
We provide an economic analysis in Section \ref{sec_economic_analysis} to study how the optimal investment strategy is affected by transaction costs and risk aversion.
The conclusions of our work are summarized in Section \ref{sec_conclusion}.

\section{The Setup}
\label{sec_model}
\subsection{The Financial Market with Transaction Costs}
\label{subsection_market}
We consider a single-period discrete-time financial market model equipped with a complete probability space $(\Omega,\F,\Po)$.
In the model, time $0$ and time $T$ ($T>0$) represent present and future, respectively.
The financial market consists of one risk-free asset and one risky asset (e.g., stock index).
Trading the risk-free asset is frictionless.
However, trading the risky asset will incur proportional transaction costs, and we denote such proportion by $\lambda$, where $\lambda \in [0,1)$.
If $\lambda = 0$, then the market is reduced to a frictionless one, as in \cite{BG10}, \cite{HZ11}, and among others.
In this paper, we are mainly concerned with the case when $\lambda > 0$, i.e., a financial market with transaction costs.

The risk-free return for the time period $[0,T]$ is $r$, where $r \ge 0$ is a constant. That means if an investor deposits \EUR{1} in the risk-free asset at time $0$, he or she will receive \euro $(1+r) $ at time $T$.

The (nominal) return on the risky asset is given by a random variable $R$.
We assume the ask price of the risky asset $S(\cdot)$ is modeled by $S(T) = (1+R) \cdot S(0)$, where $S(0)$ is a positive constant.
The bid price of the risky asset at time $t$ is given by $(1-\lambda)S(t)$, where $t=0,T$.

We assume $\F_0$ is trivial and $\F_T$ is the completion of $\sigma(S(T))$.  Thus $R$ is $\F_T$ measurable.
For any $\F_T$ measurable random variable $Z$, we denote its cumulative distribute function (CDF) by $F_Z(\cdot)$ and its survival function by $S_Z(\cdot)$.
By definition, $S_Z(\cdot)=1-F_Z(\cdot)$.

We consider an investor with initial portfolio $(x_0, y_0)$.
That means the investor starts with $x_0$ and $y_0$ amount of money in the risk-free and the risky asset, respectively.
The investor chooses the amount of money to be additionally invested in the risky asset at time $0$, denoted by $\theta$,
and carried out to terminal time $T$ when it will be liquidated.
Denote the investor's terminal wealth after liquidation by $W(\theta)$.
A straightforward computation yields
\begin{equation*}
W(\theta)=(1+r)(x_0 - \theta)+ (1+R)(y_0 + \theta) - \lambda \Big[(1+R)(y_0 + \theta)^+ + (1+r) \theta^- \Big],
\end{equation*}
where $x^+:=\max\{x,0\}$ and $x^-:=\max\{-x,0\}$ for all $x \in \R$.

In this financial market, the non-arbitrage condition reads as
\begin{equation} \label{non_arbitrage}
\Po \Big( (1-\lambda)(1 + R) < 1+r \Big) > 0, \text{ and } \Po \Big( 1 + R >  (1-\lambda)(1+r) \Big) >0.
\end{equation}

\begin{rem}
We ignore two degenerated cases: (i) $(1-\lambda)(1 + R)\equiv 1+r$, and
(ii) $1 + R \equiv (1-\lambda)(1+r)$, under which the market is also arbitrage-free.
If $\lambda=0$, meaning the market is frictionless, then the non-arbitrage condition \eqref{non_arbitrage} simply reduces to
\[ 0 < \Po( R<r) < 1 .\]
\end{rem}

\subsection{The CPT Framework}
\label{subsection_CPT}
In \cite{TK92}, Tversky and Kahneman propose cumulative prospect theory (CPT) as a performance criterion for decision making under uncertainty. A CPT model is characterized by the following three key features.
\begin{enumerate}
\item Reference point

Behavioral studies show that people do not evaluate final outcomes directly but rather compare them to some benchmark, see, e.g., \cite{TK81}.
In CPT, a reference point $B$ is chosen to serve as the benchmark for evaluating an uncertain outcome.
Let $X$ denote the final wealth of an investment decision.
If $X \ge B$, $X - B$ is considered gains from the investment; if $X<B$, $B - X$ is viewed as losses.
For example, if $B$ is set to $0$, then the terminology of gains and losses fits into the common language.

\item Utility function

Investors are not universally risk averse, instead as documented in \cite{TK92},
they exhibit a distinctive  fourfold pattern of risk attitudes towards gains and losses.
To fit those risk attitudes, CPT applies an S-shaped utility function, which consists of two different functions, $u_+$ and $u_-$, for gains and losses, respectively.
The pathwise prospect utility of an investment strategy (with associated wealth  $X$) is defined by
\[ u_+(X(\omega) - B(\omega)) \cdot \bm{1}_{X(\omega)-B (\omega)\ge 0} -  u_-(B(\omega) - X(\omega)) \cdot \bm{1}_{X(\omega)-B(\omega) < 0},\]
for all $\omega \in \Omega$, where $\bm{1}_A$ is an indicator function of set $A$.
\vspace{1em}

We assume throughout this paper that $u_{\pm}: \R^+ \to \R^+$ are twice differentiable, strictly increasing, strictly concave and satisfy $u_{\pm}(0)=0$.
Those assumptions have perfect economic and mathematical explanations, and are consistent with people's investment behaviors.
For instance, the concavity in gains (gains are evaluated by $u_+$) and the convexity in losses (losses are evaluated by $-u_-$) capture both risk aversion and risk seeking.
In addition, studies also show that people tend to prefer avoiding losses to acquiring equivalent gains, i.e., people are more sensitive to losses than to gains, see, e.g., \cite{BKP04}, \cite{HZ11}, \cite{TK81}, \cite{TK92}, and \cite{Z15}.
In economics and decision theory, such behavior is referred to as \emph{loss aversion}.
In the mathematical modeling, to capture loss aversion, we assume
\[ u_-'(x) > u_+'(x), \; \text{ for all } x \ge 0.\]

In CPT application, the most common choice for the utility function is a piece-wise power utility, first proposed by \cite{TK92}
\begin{equation} \label{utility_TK92}
u_+(x)=x^\alpha, \text{ and } u_-(x)=kx^\beta , \; \text{ for all } x \ge 0,
\end{equation}
where $k>1, \, 0<\alpha \le \beta \le 1$.
Notice that the above assumptions on the parameters are sufficient conditions for all the assumptions made on the utility function to be satisfied.
In particular,  $k>1$ and $0<\alpha \le \beta \le 1$ together imply that $ u_-'(x) > u_+'(x)$ for all $x \ge 0$, and thus the loss aversion behavior.
Examples of applying a piece-wise power utility in CPT applications can be found in \cite{BH08}, \cite{BG10}, \cite[Section 5.1]{HZ11}, \cite[Section 4.1]{PS12}, \cite{RR13}, and many others.
In \cite{TK92}, the parameters are estimated as
\begin{equation*}
\alpha=\beta=0.88, \text{ and } k=2.25,
\end{equation*}
which clearly satisfy all the parameter assumptions above.
\vspace{1em}

The power utility function, given by \eqref{utility_TK92}, is unbounded from above, and hence may lead to an ill-posed problem (either infinite CPT value or infinite optimal investment), see \cite{HZ11} for detailed discussions.
Some also argue that the power utility function may fail to explain high risk averse behavior, as pointed out in \cite{RB11}.
As a consequence, some prefer to use a piece-wise exponential utility function in CPT applications,
see arguments in \cite{KW05}.
\vspace{1em}

A piece-wise exponential utility function is given by
\begin{equation} \label{exp_utility}
u_+(x) = 1 - e^{-\eta_+ x}, \text{ and } u_-(x) = \zeta \left( 1 - e^{-\eta_- x} \right),
\end{equation}
where $\eta_+, \, \eta_->0, \zeta>1$.
In most applications, $\eta_+ = \eta_-$ is also assumed, which together with $\zeta > 1$ yields the loss aversion.
A piece-wise exponential utility function has been used in CPT related optimization problems in \cite[Section 4.3]{PS12}, \cite{RR14}, and \cite{Z15}.
\vspace{1em}



\item Probability weighting function

Investors tend to overweigh extreme events (small probability events) but underweigh normal events (large probability events).
This behavior is captured in CPT by transforming objective cumulative probabilities into subjective cumulative probabilities using the probability weighting function (also called probability distortion function).
\vspace{1em}

The weighting function has a reverse S-shape, and two separate parts for gains and losses, denoted by $w_+(\cdot)$ and $w_-(\cdot)$, respectively.
We assume that $w_{\pm}: [0,1] \to [0,1]$ are strictly increasing and differentiable, and satisfy
\[ w_{\pm}(0)=0, \text{ and } \; w_{\pm}(1)=1.\]

The weighting function used in \cite{TK92} is given by
\begin{equation} \label{weighting_TK92}
w_+(x)= \frac{x^\gamma}{(x^\gamma + (1-x)^\gamma)^{1/\gamma}}, \text{ and }
w_-(x)= \frac{x^\delta}{(x^\delta + (1-x)^\delta)^{1/\delta}}.
\end{equation}

As pointed out in \cite{RW06}, the above weighting function may fail to be strictly increasing when $\gamma, \delta \le 0.25$, but are indeed strictly increasing when $\gamma , \delta \ge 0.5$. The condition for strictly increasing weighting function is relaxed to $\gamma , \delta \ge 0.28$ in \cite{BH08}.
The estimated parameters are $\gamma = 0.61$ and $\delta=0.69$ in \cite{TK92}, which satisfy all the desired assumptions.
\vspace{1em}

In \cite{P98}, Prelec introduces the following weighting function
\begin{equation} \label{weighting_P98}
w_+(x) = e^{-\delta^+ (-\ln(x))^\gamma} \text{ and } w_-(x) = e^{-\delta^- (-\ln(x))^\gamma},
\end{equation}
where $\gamma \in (0,1)$ and $\delta^+, \delta^->0$.
\cite{RW06} use Perlec's weighting function with $\delta^+=\delta^-=1$.
\end{enumerate}
\vspace{2em}
Let $X$ be a random wealth and $B$ the reference point. We define the positive prospect $V^+(X)$ and the negative prospect $V^-(X)$ for $X$ by
\begin{align*}
V^+(X) &:= \int_B^\infty u_+(x-B) \, d[-w_+(S_X(x))], \\
V^-(X) &:= \int_{-\infty}^B u_-(B-x) \, d[w_-(F_X(x))].
\end{align*}
The prospect utility of $X$ is defined by
\[ V(X) := V^+(X) - V^-(X),\]
given that $V^+(X)$ and  $V^-(X)$ are not both infinite at the same time.

Denote $D:=X-B$, then we have
\begin{align}
\label{CPT_def1}
V(X)&= V^+(X) - V^-(X) \notag \\
&= \int_0^\infty u_+(x) d[-w_+(S_D(x))] - \int_{-\infty}^0 u_-(-x) d[w_-(F_D(x))]\\
:&=V_D(D). \notag
\end{align}
Through integration by parts and change of variable, we rewrite $V_D(D)$ in \eqref{CPT_def1} as
\begin{equation}
\label{CPT_def2}
V_D(D)=\int_0^\infty  w_+ \big(S_D(x) \big) du_+(x)  - \int_0^{\infty} w_- \big(F_D(-x) \big)du_-(x).
\end{equation}

\subsection{The Problem}
\label{subsection_problem}
In the financial market described in Subsection \ref{subsection_market}, an investor selects investment strategy $\theta$ under the CPT framework introduced in Subsection \ref{subsection_CPT}.
In other words, the investor wants to maximize the prospect utility $V_D(W(\theta) - B)$, which is defined by \eqref{CPT_def1} or \eqref{CPT_def2}.
The investor's terminal wealth $W(\theta)$ is a function of the investment strategy $\theta$, so is the prospect utility $V_D(W(\theta) - B)$. We then denote
\[ J(\theta):=V_D(W(\theta) - B).\]

The reference point $B$ is given by
\begin{equation} \label{reference_point}
 B= W(0)=(1+r)x_0 + (1+R) \left(y_0 - \lambda y_0^+ \right),
\end{equation}
i.e., the reference point is the terminal wealth of the ``doing nothing strategy".
The above selection on the reference point can also be seen in \cite{BG10}, \cite[Section 5.1]{HZ11}, and \cite{PS12}.

In the above market setting, we implicitly assume that the investor we consider is a ``small investor", and his/her investment activities do not have any impact on the price of the risky asset.

$J(\theta)$ is finite if the following assumption holds.
\begin{asu} \label{assumption_finiteness_CPT}
We assume both $V^+(W(\theta))$ and $V^-(W(\theta))$ are finite for all $\theta \in \R$.
\end{asu}

\begin{prop} \label{prop_finiteness_CPT}
Assumption \ref{assumption_finiteness_CPT} is satisfied if one of the following conditions holds.
\begin{itemize}
  \item The risky return $R$ is bounded, e.g., $R$ is a discrete random variable and $|R| \neq \infty$.
  \item The risky return $R$ follows a normal distribution, log-normal distribution, or student-t distribution, and for $x$ small enough, there exists some $ 0<\epsilon<1$ such that
  \[ w_\pm'(x)= O \left(x^{-\epsilon} \right), \text{ and } w_\pm'(1-x)= O \left(x^{-\epsilon} \right).\]
\end{itemize}
\end{prop}
\begin{proof} The first result is obvious. For the proof of the second result,
please refer to \cite[Proposition 1]{HZ11} and \cite[Proposition 2.1]{PS12}.
\qed
\end{proof}

We then formulate optimal investment problems with transaction costs under CPT as follows.

\begin{prob} \label{main_problem}
In a financial market with transaction costs (as modeled in Subsection \ref{subsection_market}), an investor seeks the optimal investment strategy to maximize the prospect utility $J(\theta)$ of his/her terminal wealth. Equivalently, the investor seeks the maximizer $\theta^*$ to the problem
\[ J(\theta^*)=\sup_{\theta \in \R} \; J(\theta) = \sup_{\theta \in \R} \; V_D(W(\theta) - B). \]
\end{prob}

\section{Explicit Solution When $R$ Has a Continuous Distribution}
\label{section_case1}
In this section, we consider the case in which the risky return $R$ has a continuous distribution.
To obtain explicit solutions to Problem \ref{main_problem}, we assume all the assumptions below hold in this section.

\begin{asu} \label{assumption_case1}
\quad

\begin{enumerate}
  \item The initial position on the risky asset is positive, $y_0 > 0$.
  \item Short-selling is not allowed, i.e., $\theta \ge - y_0 $.
  \item The utility function is of power type, given by \eqref{utility_TK92}.
\end{enumerate}
\end{asu}

\begin{rem} We make some comments on Assumption \ref{assumption_case1}.
We do allow investors to sell the risky asset ($\theta$ can be negative), but no more than what they currently own.
However, the no short-selling constraint imposed in \cite{BG10} is equivalent to $\theta \ge 0$.
The case of $y_0 < 0$ is less interesting since the no short-selling constraint then implies $\theta > 0$.
\end{rem}

\subsection{Main Results}
To find the optimal solution to Problem \ref{main_problem}, we separate the prime problem $\sup_{\theta \in \R} J(\theta)$ into two sub-problems:
\[\text{(P1)}  \;  \sup_{\theta \ge 0} \, J(\theta), \text{ and } \text{(P2)} \;
\sup_{\theta \le 0} J(\theta) . \]
By comparing the optimal prospect utility of the two sub-problems (P1) and (P2), we obtain the optimal investment strategy to Problem \ref{main_problem}.
Since in Assumption \ref{assumption_case1}, we impost the no short-selling constraint, i.e., $\theta \ge -y_0$, sub-problem (P2) is reduced to
\[\text{(P2)} \; \sup_{-y_0 \le \theta \le 0} J(\theta) . \]

Denote $A_1$ and $A_2$ as the set of losses for all ``buy strategies'' ($\theta \ge 0$) and all ``sell strategies'' ($\theta \le 0$), respectively.
Namely, if $\omega \in A_1$ (or $\omega \in A_2$), then $D = W(\theta) - B < 0$ for all $\theta \ge 0$ (or for all $\theta \le 0$).
The formal mathematical definitions of $A_1$ and $A_2$ are deferred in the corresponding subsections.
Define
\[ Z_1:= (1-\lambda)(1+R) - (1+r),\;\;\text{and} \;\; Z_2:= (1-\lambda) (R-r),\]
as well as
\begin{equation}\label{def_theta1_theta2}
   \theta_1:= \left( \frac{\alpha}{\beta  k} K_1(Z_1) \right) ^{\frac{1}{\beta - \alpha} }, \;\;\text{and} \;\;
   \theta_2:= - \left( \frac{\alpha}{\beta  k} K_2(Z_2) \right) ^{\frac{1}{\beta - \alpha} },
\end{equation}
where $K_1(Z_1)$ and $K_2(Z_2)$ are defined in the sequel by \eqref{eqn_K1} and \eqref{def_K2}.
Denote
\[ K_M = \max\{ K_1(Z_1), K_2(Z_2)\}.\]

We then summarize the main results of this section in the theorem below.

\begin{thm} \label{thm_1case}
If Assumption \ref{assumption_case1} holds, we have the following results for the optimal investment $\theta^*$ to Problem \ref{main_problem}.
\begin{enumerate}
  \item $\theta^* = 0$ if one of the following conditions holds:
  \begin{enumerate}
    \item $\Po (A_1) =\Po (A_2) =1$;
    \item $\Po (A_1) =1$, $0 < \Po (A_2) <1$, $\alpha = \beta$, and $k>K_2(Z_2)$;
    \item $0<\Po (A_1)<1$, $\Po (A_2) =1$, $\alpha = \beta$, and $k>K_1(Z_1)$;
     \item $0<\Po (A_1),\Po (A_2) <1$, $\alpha = \beta$, and $k>K_M$.
  \end{enumerate}
  \vspace{0.5em}

  \item $\theta^* = \theta_1$ if one of the following conditions holds:
  \begin{enumerate}
    \item $0<\Po (A_1)<1$, $\Po (A_2) =1$, and $\alpha < \beta$;
    \item $0<\Po (A_1),\Po (A_2) <1$, $\alpha < \beta$, and $J(\theta_1) \ge \max\{J(\theta_2), J(-y_0) \}$.
  \end{enumerate}
  \vspace{0.5em}

  \item $\theta^* = \theta_2$ if one of the following conditions holds:
  \begin{enumerate}
    \item $\Po (A_1)=1$, $0<\Po (A_2) <1$, $\alpha < \beta$, and $-y_0 \le \theta_2$;
    \item $0<\Po (A_1),\Po (A_2) <1$, $\alpha < \beta$, $-y_0 \le \theta_2$, and $J(\theta_2) \ge J(\theta_1)$.
  \end{enumerate}
  \vspace{0.5em}

  \item $\theta^* = -y_0$ if one of the following conditions holds:
  \begin{enumerate}
    \item $\Po(A_1)=1$ and $\Po(A_2) =0$;
    \item $\Po (A_1)=1$, $0<\Po (A_2) <1$, $\alpha = \beta$, and $k < K_2(Z_2)$;
    \item $\Po (A_1)=1$, $0<\Po (A_2) <1$, $\alpha < \beta$, and $-y_0 > \theta_2$;
    \item $0<\Po (A_1),\Po (A_2) <1$, $\alpha < \beta$, $-y_0 \ge \theta_2$, and $J(-y_0) \ge J(\theta_1)$.
  \end{enumerate}
  \vspace{0.5em}

  \item Any $\theta^* \in [0, +\infty)$ is an optimal strategy if one of the following conditions holds:
  \begin{enumerate}
    \item $0<\Po (A_1)<1$, $\Po (A_2) =1$, $\alpha = \beta$, and $k = K_1(Z_1)$;
    \item $0<\Po (A_1),\Po (A_2) <1$, $\alpha = \beta$, and $k = K_1(Z_1) > K_2(Z_2)$.
  \end{enumerate}
  \vspace{0.5em}

  \item Any $\theta^* \in [-y_0, 0]$ is an optimal strategy if one of the following conditions holds:
  \begin{enumerate}
    \item $\Po (A_1)=1$, $0<\Po (A_2) <1$, $\alpha = \beta$, and $k = K_2(Z_2)$;
    \item $0<\Po (A_1),\Po (A_2) <1$, $\alpha = \beta$, and $k = K_2(Z_2) > K_1(Z_1)$.
  \end{enumerate}
  \vspace{0.5em}

  \item Any $\theta^* \in [-y_0, \infty)$ is an optimal strategy if the following condition holds:

  $0<\Po (A_1),\Po (A_2) <1$, $\alpha = \beta$, and $k = K_1(Z_1) = K_2(Z_2)$.
  \vspace{0.5em}

  \item $\theta^* = +\infty$ if one of the following conditions holds:
  \begin{enumerate}
    \item $0<\Po (A_1)<1$, $\Po (A_2) =1$, $\alpha = \beta$, and $k<K_1(Z_1)$;
    \item $0<\Po (A_1),\Po (A_2) <1$, $\alpha = \beta$, and $k<K_M$.
  \end{enumerate}
\end{enumerate}
\end{thm}

The rigorous proof of Theorem \ref{thm_1case} is an immediate observation once we have solved the two sub-problems (P1) and (P2) in Subsection \ref{subsection_1case_p1} and Subsection \ref{subsection_1case_p2}.
We provide some economic meanings and incomplete mathematical reasonings for some representative cases in Theorem \ref{thm_1case} (the rest can be interpreted in a similar way),
which readers may find helpful to understand Theorem \ref{thm_1case}.

\begin{rem} \label{rem_case1}
\begin{itemize}
  \item $\Po(A_1) = 1$ (alternatively, $\Po(A_2) = 1$) means that all long \mbox{strategies} (alternatively, all short strategies) result in losses almost surely.
  Hence, if both $\Po(A_1) = 1$ and $\Po(A_2) = 1$, the optimal investment strategy is obviously zero, i.e., $\theta^*=0$ as in Case (1a).

  \item In (P1), we have sign$(J'(\theta))$ = sign$(K_1(Z_1) - k)$ when $\alpha = \beta$.
  In (P2), we have sign$(J'(\theta))$ = sign$(K_2(Z_2) - k)$ when $\alpha = \beta$ (both are shown later).
  \begin{itemize}
    \item If $\Po(A_1) = 1$, as a consequence of the analysis above, the optimal investment strategy will be a short strategy ($\theta^* \in [-y_0 ,0]$).
    In the case when $\alpha = \beta$ and $k < K_2(Z_2)$, $J'(\theta) <0$,  and $J(\theta)$ is a decreasing function, implying that $\theta^* = -y_0$, which is exactly the result of Case (4b).
    The same analysis also applies to Case (1b) and Case (6a).

    \item We directly obtain the results from symmetric cases when $\Po(A_2) = 1$ and $\alpha = \beta$, as in Case(1c), Case (5a), and Case (8a)
  \end{itemize}

  \item In all non-trivial cases (i.e., investors may end up in gains or losses with strictly positive probabilities) and when $\alpha<\beta$, $J(\theta)$ has a unique maximizer $\theta_1>0$ in (P1) and a unique maximizer $\theta_2<0$ in (P2).
  In (P1), the feasible region is unbounded, so $\theta_1$ is achievable. But in (P2), the feasible region is bounded below by $-y_0$, so $\theta_2$ is the maximizer if and only if $\theta_2 \ge -y_0$; such extra condition can be seen in Case (3a), Case (3b), Case (4c), and Case (4d).
  The optimal investment strategy is then obtained by comparing $J(\theta_1)$ and $J(\theta_2)$ (or $J(-y_0)$), see, for instance, Case (2b) and Case (3b).

  If the constraint $\theta \ge -y_0$ is not binding in the above situation (i.e., $0<\Po(A_1), \Po(A_2) <1$ and $\alpha < \beta$), we obtain finer results for Case (2b) and Case (3b):
\begin{itemize}
\item[(i)] If $\alpha<\beta$ and $(g_1(Z_1))^\beta / (l_1(Z_1))^\alpha \ge (g_2(Z_2))^\beta / (l_2(Z_2))^\alpha$, then $\theta^* = \theta_1$.

\item[(ii)] If $\alpha<\beta$ and $(g_1(Z_1))^\beta / (l_1(Z_1))^\alpha < (g_2(Z_2))^\beta / (l_2(Z_2))^\alpha$, then $\theta^* = \theta_2$.
\end{itemize}
The above results are based on the comparison between $J(\theta_1)$ and $J(\theta_2)$, see \cite[Appendix]{HZ11}.

\item The optimal investment $\theta^*$ and the optimal prospect $J(\theta^*)$ are both finite in all cases except Case (8), in which $\theta^* = +\infty$.
If either of the conditions in Case (8) is satisfied, Problem \ref{main_problem} is ill-posed.
Please see \cite[Section 3]{HZ11} for more discussions on the well posedness of the problem.
\end{itemize}
\end{rem}

\subsection{Solution to Sub-Problem (P1)}
\label{subsection_1case_p1}
If $\theta \ge 0$, then $y_0 + \theta \ge 0$. Hence the investor needs to sell all the holdings in the risky asset at liquidation.

Recall $Z_1 = (1-\lambda) (1+R) - (1 +r)$, we obtain
\[ D=W(\theta) - B = Z_1 \cdot \theta.\]

Define set $A_1$ by
$A_1:=\{ Z_1 <0\} = \left\{ 1+R < \frac{1+r}{1- \lambda} \right\}$ .
Notice that set $A_1$ is the set of losses for the investor (recall $\theta \ge 0$). Due to the non-arbitrage condition \eqref{non_arbitrage},  $\Po(A_1)>0$.

By definition \eqref{CPT_def1} and change of variable ($x = z \theta$, $\theta>0$), the prospect utility $J(\theta)$ in sub-problem (P1) can be rewritten as
\begin{align*}
J(\theta) &= \int_0^\infty x^\alpha d[-w_+(S_D(x))] - \int_{-\infty}^0 k(-x)^\beta d[w_-(F_D(x))] \\
&= \int_0^\infty z^\alpha d\left[ -w_+ \big(S_{Z_1}(z) \big) \right] \cdot \theta^\alpha -\int_{-\infty}^0 (-z)^\beta d\left[w_- \big(F_{Z_1}(z) \big) \right] \cdot k \theta^\beta.
\end{align*}

We define, for any $\F_T$ measurable random variable $Z$, that
\begin{equation} \label{def_g1_l1}
\begin{split}
g_1(Z):&= \int_0^\infty z^\alpha d\left[ -w_+ \big(S_{Z}(z) \big) \right],\\
l_1(Z):&=\int_{-\infty}^0 (-z)^\beta d\left[w_- \big(F_{Z}(z) \big) \right].
\end{split}
\end{equation}
In general, $g_1(Z)$ (or $l_1(Z)$) can be understood as the prospect value of gains (or losses) of random wealth $X$ with $X-B=Z$.
In our setting here, $g_1(Z_1)$ and $l_1(Z_1)$ are exactly the prospect value of gains and the prospect value of losses (differ by a scalar $k$) of $W(1)$, which is the terminal wealth associated with the strategy $\theta=1$. Mathematically, we have $g_1(Z_1) = V^+(W(1))$ and $k \cdot  l_1(Z_1)=V^-(W(1))$.

With the definitions of $g_1$ and $l_1$, the prospect utility $J(\theta)$ is simplified as
\begin{equation*}
J(\theta) = g_1(Z_1) \cdot \theta^\alpha - l_1(Z_1) \cdot k \theta^\beta, \; \theta\ge 0.
\end{equation*}

Define $K_1(Z_1)$ by
\begin{equation}\label{eqn_K1}
  K_1(Z_1) : = \dfrac{g_1(Z_1)}{l_1(Z_1)} .
\end{equation}
Since $\Po(A_1)>0$, $K_1(Z_1)$ is well defined and $K_1(Z_1) > 0$.
Notice that $K_1(Z)$ shares similar features as the \emph{Omega Measure} proposed by \cite{KS02}, see \cite[Section 4.1]{BG10} for comparisons.

We summarize the solution to sub-problem (P1) below.

\begin{thm} \label{thm_1case_positive}
If Assumption \ref{assumption_case1} holds, then the optimal solution $\theta^*$ to sub-problem (P1) is obtained from one of the following scenarios.

\begin{enumerate}
  \item $\theta^*=0$ if either (a) $\Po(A_1)= 1$ or (b) $0 < \Po(A_1) <1$, $\alpha=\beta$, and $k > K_1(Z_1)$ holds.

  \item $\theta^*= \theta_1$, given by \eqref{def_theta1_theta2}, if $0 < \Po(A_1) <1$ and $\alpha<\beta$.

  \item Any $\theta^* \in [0, +\infty)$ is an optimal solution if $0 < \Po(A_1) <1$, $\alpha=\beta$, and $k = K_1(Z_1)$.

  \item $\theta^*=+\infty$ if $0 < \Po(A_1) <1$, $\alpha=\beta$, and $k < K_1(Z_1)$.
\end{enumerate}
\end{thm}

\begin{proof} If $\Po(A_1)=1$, then the probability of suffering losses is $1$ for all long strategies $\theta \ge 0$. Thus it is never optimal to buy the risky asset, i.e., $\theta^*=0$.
Mathematically, $\Po(A_1)=\Po(Z_1<0)=1 \Rightarrow g_1(Z_1)=0$.
Then we have
\[ J(\theta) = -l_1(Z_1) \cdot k\theta^\beta < J(0)=0, \text{ for all }\theta >0,\]
which directly indicates $\theta^*=0$.

We next consider the non-trivial case: $0 < \Po(A_1) <1$.

Differentiating $J(\theta)$ gives
\[ J'(\theta) = l_1(Z_1) \, \theta^{\alpha - 1} \left[ \alpha \cdot K_1(Z_1) -  \beta k \cdot \theta^{\beta - \alpha} \right] .\]

If $\alpha = \beta$,  we have sign$(J'(\theta))$ = sign$(K_1(Z_1) - k)$.
Hence, $J(\theta)$ is either a strictly decreasing or strictly increasing function or a constant, depending on the value of $k$.
If the condition in Case (4) is satisfied, $J(\theta)$ is strictly increasing and $\lim_{\theta \to \infty} J(\theta) = +\infty$. Thus (P1) is an ill-posed problem.

If $\alpha<\beta$, then $\theta_1$, defined by \eqref{def_theta1_theta2}, is the unique solution to $J'(\theta)=0$ on the positive axis.
Furthermore, $J'(\theta)>0$ for all $\theta \in (0,\theta_1)$ and $J'(\theta)<0$ for all $\theta \in (\theta_1, \infty)$.
Therefore, $\theta_1$ is the unique maximizer to (P1).
\qed
\end{proof}

\subsection{Solution to Sub-Problem (P2)}
\label{subsection_1case_p2}
Due to the no short-selling constraint in Assumption \ref{assumption_case1}, we have $y_0 + \theta \ge 0$ for all $-y_0 \le \theta \le 0$. So the liquidation order at terminal time $T$ is to sell all the risky assets.

Recall $Z_2=(1-\lambda)(R-r)$, we obtain in this case that
\[ D=W(\theta) - B = Z_2 \cdot \theta .\]

Define set $A_2$ by
$ A_2: = \{Z_2 > 0\} = \left\{ R>r \right\}.$
Since investment $\theta$ is restricted to short strategies ($\theta \le 0$) in this subsection, the difference $D$ is negative on set $A_2$, meaning that set $A_2$ is the set of losses for the investor.

In this case, we obtain $J(\theta)$ as follows:
\begin{align*}
J(\theta) &= \int_0^\infty x^\alpha d[-w_+(S_D(x))] - \int_{-\infty}^0 k(-x)^\beta d[w_-(F_D(x))] \\
&= \int_{-\infty}^0 (-z)^\alpha d\left[ w_+ \big(F_{Z_2}(z) \big) \right] (-\theta)^\alpha -
\int_0^\infty z^\beta d\left[- w_- \big(S_{Z_2}(z) \big) \right] \cdot k (-\theta)^\beta,
\end{align*}
where we have applied the change of variable $ x = z\theta$ ($\theta<0$) in the second equality.

We define, for any $\F_T$ measurable random variable $Z$, that
\begin{equation} \label{def_g2_l2}
\begin{split}
g_2(Z):&= \int_{-\infty}^0 (-z)^\alpha d\left[ w_+ \big(F_{Z}(z) \big) \right],\\
l_2(Z):&= \int_0^\infty (z)^\beta d\left[- w_- \big(S_{Z}(z) \big) \right].
\end{split}
\end{equation}
The economic meanings of $g_2(Z)$ and $l_2(Z)$ are similar to those of $g_1(Z)$ and $l_1(Z)$, except the gains/losses are located on exactly opposite tails due to
the different signs of $\theta$ in two cases.

Define $ K_2(Z_2)$ by
\begin{equation}\label{def_K2}
  K_2(Z_2):=\frac{g_2(Z_2)}{l_2(Z_2)},
\end{equation}
which is well defined and strictly positive if $\Po(A_2)>0$.

Using the notations of $g_2(\cdot)$ and $l_2(\cdot)$, we rewrite  $J(\theta)$ as
\begin{equation*}
J(\theta) = g_2(Z_2) \cdot (-\theta)^\alpha -  l_2(Z_2) \cdot k (-\theta)^\beta, \; \text{ for } -y_0 \le \theta \le 0.
\end{equation*}
The unique solution to $J'(\theta)=0$ on the negative axis is $\theta_2$, given by \eqref{def_theta1_theta2}.

We directly provide the results to sub-problem (P2). Please refer to Theorem \ref{thm_1case_positive} for a similar proof.

\begin{thm} \label{thm_1case_negative}
If Assumption \ref{assumption_case1} holds, then the optimal solution $\theta^*$ to sub-problem (P2) is obtained from one of the following scenarios.

\begin{enumerate}
  \item $\theta^*=0$ if either (a) $\Po(A_2)= 1$ or (b) $0 < \Po(A_2) <1$, $\alpha=\beta$, and $k > K_2(Z_2)$ holds.

  \item $\theta^*= \theta_2$, given by \eqref{def_theta1_theta2}, if $0 < \Po(A_2) <1$, $\alpha<\beta$, and $\theta_2 \ge -y_0$.

  \item Any $\theta^* \in [-y_0, 0]$ is an optimal solution if $0 < \Po(A_2) <1$, $\alpha=\beta$, and $k = K_2(Z_2)$.

  \item $\theta^*=-y_0$ if one of the following conditions holds:
  \begin{enumerate}
    \item $\Po(A_2)= 0$;
    \item $0 < \Po(A_2) <1$, $\alpha=\beta$, and $k < K_2(Z_2)$;
    \item $0 < \Po(A_2) <1$, $\alpha<\beta$, and $\theta_2 < -y_0$.
  \end{enumerate}

\end{enumerate}
\end{thm}

\begin{rem}
Notice that $\Po(A_1)$ is strictly positive but $\Po(A_2)$ may be zero. Moreover, if $\Po(A_2) = 0$, then $\Po(A_1)=1$.
This finding helps reduce the cases in Theorem \ref{thm_1case}.
We also observe that the optimal solution in Case (4) will change to $\theta^* = -\infty$ if the no short-selling constraint is dropped.
\end{rem}

\subsection{Discussions for $y_0=0$}
To obtain the conclusions in Theorem \ref{thm_1case}, we suppose $y_0 > 0$ in Assumption \ref{assumption_case1}.
If the investor does not hold any risky asset at time $0$ ($y_0=0$), we can remove the constraint of no short-selling and still obtain explicit solutions.
Notice that $B=(1+r)x_0$ when $y_0=0$, which is the most common choice for the reference point and is used by
\cite{BG10}, \cite{HZ11}, \cite{PS12}, and many others.

The solution to sub-problem (P1) is exactly the same as in Theorem \ref{thm_1case_positive}. However, the solution to sub-problem (P2) here is different from the results in Theorem \ref{thm_1case_negative}.
Given $y_0=0$, we have $y_0 + \theta \le 0$ for all $\theta \le 0$ (recall $y_0 + \theta \ge 0$ in Subsection \ref{subsection_1case_p2}).

Define $Z_3:=1+R - (1-\lambda)(1+r)$. Then we have
\[ D=W(\theta) - B = Z_3 \cdot \theta \quad \text{for all } \theta \le 0.\]
Define the set of losses $A_3$ by
$A_3:=\{ Z_3 > 0\} = \{1+R > (1-\lambda)(1+r) \}$.
The non-arbitrage condition \eqref{non_arbitrage} implies $\Po(A_3) >0$, while $\Po(A_2)=0$ is possible in Subsection \ref{subsection_1case_p2}.

By replacing $Z_2$ by $Z_3$, $A_2$ by $A_3$ and removing Case (4a) in Theorem \ref{thm_1case_positive}, we solve (P2) in the case of $y_0=0$.
Then a similar theorem as Theorem \ref{thm_1case} can be obtained easily for the case of $y_0 = 0$ without any constraint on investment strategies.

To study the connection between those two cases ($y_0>0$ and $y_0=0$), we modify the notations for $J(\theta)$.
For initial position $(x_0, y_0)$ with $y_0>0$, we use $\bar{J}(\theta; x_0, y_0)$ to replace $J(\theta)$.
If we fix $x_0$ and $\theta$ and treat $\bar{J}$ as a function of $y_0$, then $\bar{J}$ becomes a continuous function in argument $y_0$.
We then extend the definition of $\bar{J}$ by
\[ \bar{J}(\theta; x_0, 0) =\lim_{y_0 \to 0} \bar{J}(\theta; x_0, y_0) \quad \text{for all } x_0 \in \R, \, \theta \ge -y_0.\]
Denote the optimal strategy by $\bar{\theta}^*$ when $y_0>0$ and the no short-selling constraint is imposed, i.e.,
\[ \bar{\theta}^*(x_0, y_0) = \argmax_{\theta \ge -y_0} \bar{J}(\theta; x_0, y_0).\]
Despite the notation difference, we clearly have $\bar{\theta}^* = \theta^*$ given by Theorem \ref{thm_1case}.
If the investor does not hold any risky asset at time $0$ (i.e., $y_0 = 0$),
we use $\tilde{J}(\theta; x_0)$ to replace $J(\theta)$, where $\theta \in \R$.
Under those new notations,
\[ \tilde{J}(\theta; x_0) = \bar{J}(\theta; x_0, 0) \quad \text{for all } x_0 \in \R, \, \theta \ge -y_0.\]
If we extend the definition $\bar{J}$ from $\theta \ge -y_0$ to $\theta \in \R$, then the above equality holds on the entire real line for $\theta$.
Denote the optimal strategy by $\tilde{\theta}^*$ when $y_0=0$ and the no short-selling constraint is \emph{not} imposed, i.e.,
\[ \tilde{\theta}^*(y_0) = \argmax_{\theta \in \R} \tilde{J}(\theta; x_0).\]
We have the following proposition regarding the optimal prospect of those two cases.

\begin{prop}
Given $x_0 \in \R$ and $y_0 >0$, denote $x_{0,1} := x_0 +y_0 $ and $x_{0,2} := x_0 + (1-\lambda)y_0$.
The inequality
\[ \bar{J}(\bar{\theta}^*(x_0, y_0); x_0, y_0) \le \tilde{J}(\tilde{\theta}^*(x_{0,1}); x_{0,1}) \]
holds if neither of the two conditions in Case (8) of Theorem \ref{thm_1case} is satisfied.

The inequality
\[ \bar{J}(\bar{\theta}^*(x_0, y_0); x_0, y_0) \ge \tilde{J}(\tilde{\theta}^*(x_{0,2}); x_{0,2})\]
holds if $ 0 \le \tilde{\theta}^*(x_{0,2}) < \infty$.
\end{prop}
\begin{proof}
We first observe that $\bar{J}(\bar{\theta}^*(x_0, y_0); x_0, y_0) < +\infty$ if and only if  neither of the two conditions in Case (8) of Theorem \ref{thm_1case} is satisfied.
Consider investor $\mathcal{I}_1$ with initial portfolio $(x_0,y_0)$ and investor $\mathcal{I}_2$ with $x_{0,1}$ in the riskless asset.
No matter what strategy $\theta$ investor $\mathcal{I}_1$ chooses, investor $\mathcal{I}_2$ can always choose $\theta' = \theta + y_0$ to outperform investor $\mathcal{I}_1$.
To see this fact, we discuss two scenarios:
\begin{itemize}
  \item $\theta \ge 0$

  At time $0$ right after the investment decision, the portfolios of both investor $\mathcal{I}_1$ and investor $\mathcal{I}_2$ become $(x_0 - \theta, y_0 + \theta)$.

  \item $-y_0 \le \theta \le 0$

  At time $0$ right after the investment decision, the portfolio of investor $\mathcal{I}_1$ is now $(x_0 - (1-\lambda)\theta, y_0 + \theta)$ while the portfolio of investor $\mathcal{I}_2$ is still $(x_0 - \theta, y_0 + \theta)$. Notice $x_0 - \theta \ge x_0 - (1-\lambda)\theta$.
\end{itemize}
The above fact obviates the first inequality.

If $ 0 \le \tilde{\theta}^*(x_{0,2}) < \infty$, then $\tilde{J}(\tilde{\theta}^*(x_{0,2}); x_{0,2})<+\infty$ and $\theta' = \tilde{\theta}^*(x_{0,2}) - y_0 \ge -y_0$ is a feasible investment strategy for an investor with initial portfolio $(x_0,y_0)$. Following a similar argument as in the proof of the first inequality leads to
$\bar{J}(\theta'; x_0, y_0) \ge \tilde{J}(\tilde{\theta}^*(x_{0,2}); x_{0,2})$, implying the second inequality.
\qed
\end{proof}

\subsection{Discussions for $\lambda=0$}
In the model setup, we allow $\lambda \ge 0$. The analysis and results in this section so far are more interesting when $\lambda>0$, since similar results are obtained in a frictionless market (corresponding to $\lambda=0$), see, e.g., \cite{BG10}, \cite{HZ11}, and \cite{PS12}.
In the rest of this section, we conduct some comparisons for two cases when $\lambda>0$ and $\lambda=0$.
\begin{enumerate}
  \item If $\lambda>0$, $Z_1=(1-\lambda)(1+R) - (1+r) < Z_2 = (1-\lambda)(R-r)$.
  If $\lambda = 0$, $Z_1 = Z_2 = R-r$.

  \item Recall $A_1 = \{Z_1<0\}$ and $A_2 = \{Z_2>0\}$.
  Given non-arbitrage condition \eqref{non_arbitrage}, $0 < \Po(A_1) \le 1$ and $ 0 \le \Po(A_2) \le 1$ if $\lambda>0$; and
  $0 < \Po(A_1) < 1$ and $ 0 < \Po(A_2) < 1$ if $\lambda=0$.

  \item If $\lambda = 0$, then the results of Theorem \ref{thm_1case} are reduced to those in a frictionless market (see, for instance, \cite[Theorem 3.1]{BG10} and \cite[Theorem 3]{HZ11}).
  In addition, as a result of the previous comparison, several Cases in Theorem \ref{thm_1case} will never happen, these cases include (1a)-(1c), (2a), (3a), (4a)-(4c), (5a), (6a), and (8a).

  \item Using the CPT definition \eqref{CPT_def2}, we rewrite $g_i(Z_j)$, $i,j=1,2$, defined in \eqref{def_g1_l1} and \eqref{def_g2_l2}, as
\begin{align*}
g_1(Z_1) &= \int_0^{\infty} w_+(S_{Z_1}(z)) du_+(z), & l_1(Z_1)&=\frac{1}{k} \int_0^{\infty} w_-(F_{Z_1}(-z)) du_-(z), \\
g_2(Z_2) &= \int_0^{\infty} w_+(F_{Z_2}(-z)) du_+(z), & l_2(Z_2)&=\frac{1}{k} \int_0^{\infty} w_-(S_{Z_2}(z)) du_-(z).
\end{align*}

If $\lambda>0$, we have $Z_2 > Z_1$, which implies $F_{Z_2}(z) \le F_{Z_1}(z)$ for all $z$ (strict inequality holds for some $z$).
Furthermore, if $Z_2$ is symmetrically distributed around $0$ (equivalently, $R$ is symmetrically distributed around $r$), we have, $\forall \, z>0$
\begin{align*}
 F_{Z_2}(-z) &= 1 - F_{Z_2}(z) \ge 1 - F_{Z_1}(z) = S_{Z_1}(z), \\
 F_{Z_1}(-z) &= 1 - S_{Z_1}(-z) \ge 1 - S_{Z_2}(-z)= S_{Z_2}(z).
\end{align*}
Therefore $g_2(Z_2) > g_1(Z_1)$ and $l_2(Z_2) < l_1(Z_1)$.
Consequently, $K_2(Z_2) > K_1(Z_1)$ holds, and then $K_M=K_2(Z_2)$.

If $\lambda =0$, and $Z_2$ is symmetrically distributed around $0$, we have $K_1(Z_1) = K_2(Z_2)$.
\end{enumerate}

\section{Explicit Solution When $R$ Has a Binomial Distribution}
\label{sec_case2}
In this section, we solve Problem \ref{main_problem} explicitly when the return of the risky asset $R$ has a binomial distribution, specified by
\begin{equation} \label{binomial_model}
1 + R =\begin{cases}
u, & \text{ with probability } 1-p \\
d,  & \text{ with probability } p
\end{cases},
\end{equation}
where $u>d>0$ and $0<p<1$.

The non-arbitrage condition \eqref{non_arbitrage} in this model reads as
\[ u > (1-\lambda) (1+r) > (1-\lambda)^2 d .\]

Given a payoff $\xi \in \F_T$, assume
\begin{equation*}
\xi =\begin{cases}
\xi_u, & \text{ when } 1+ R= u \\
\xi_d,  & \text{ when } 1+ R= d
\end{cases}.
\end{equation*}
In what follows, we may denote $\xi=(\xi_u, \xi_d)$ in the above sense.
In the market modeled by \eqref{binomial_model}, assume we can replicate $\xi$ by strategy $\theta_\xi$ and initial investment $x_\xi$.

\begin{itemize}
\item If $\xi_u \ge \xi_d$, then we obtain\footnote{In this case, the replication strategy involves long the risky asset.
$\theta_\xi$ and $x_\xi$ are solved from
$
(1+r) \cdot (x_\xi - \theta_\xi) + (1-\lambda)u \cdot \theta_\xi = \xi_u \text{ and }
(1+r) \cdot (x_\xi - \theta_\xi) + (1-\lambda)d \cdot  \theta_\xi = \xi_d.
$}
\begin{align}
\theta_\xi &= \frac{\xi_u - \xi_d}{(1-\lambda)(u-d)}, \label{replication_strategy1}\\
x_\xi &= \frac{1}{1+r} \left(p^b_u \cdot \xi_u + p^b_d \cdot \xi_d \right),\label{replication_cost1}
\end{align}
where $p^b_u$ and $p^b_d$ are defined by
\begin{equation*}
p^b_u:=\frac{(1+r) - (1-\lambda)d}{(1-\lambda)(u-d)}, \;\; \text{and} \;\;
p^b_d:=\frac{(1-\lambda)u - (1+r)}{(1-\lambda)(u-d)}.
\end{equation*}

\item If $\xi_u < \xi_d$, then we obtain\footnote{In this case, the replication strategy involves short the risky asset.
$\theta_\xi$ and $x_\xi$ are solved from
$
(1+r) \cdot (x_\xi - (1-\lambda)\theta_\xi) + u \cdot \theta_\xi = \xi_u \text{ and }
(1+r) \cdot (x_\xi - (1-\lambda)\theta_\xi) + d \cdot  \theta_\xi = \xi_d.
$}
\begin{align}
\theta_\xi &= \frac{\xi_u - \xi_d}{u-d}, \label{replication_strategy2}\\
x_\xi &= \frac{1}{1+r} \left(p^s_u \cdot \xi_u + p^s_d \cdot \xi_d \right),\label{replication_cost2}
\end{align}
where $p^s_u$ and $p^s_d$ are defined by
\begin{equation*}
p^s_u:=\frac{(1-\lambda)(1+r) - d}{u-d}, \;\;\text{and} \;\;
p^s_d:=\frac{u- (1-\lambda)(1+r)}{u-d}.
\end{equation*}
\end{itemize}

\begin{rem}
If $\xi_u \ge \xi_d$ (or $\xi_u < \xi_d$), the replication strategy involves buying (or selling) the risky asset (since $\theta_\xi \ge 0$ in \eqref{replication_strategy1} and $\theta_\xi<0$ in \eqref{replication_strategy2}).

Notice that $p^b_u+ p^b_d=1$, $p^s_u+ p^s_d=1$ and $p^b_u, \, p^s_d>0$, but $p^b_d$ and $p^s_u$ may be negative, so $(p^b_u, p^b_d)$ and $(p^s_u, p^s_d)$ are not necessarily risk-neutral probability measures. However, if $\lambda=0$, we have $p^b_u=p^s_u$, $p^b_d=p^s_d$, and $(p^b_u, p^b_d)$ is indeed the unique risk-neutral probability measure.

\end{rem}

To solve Problem \ref{main_problem}, we claim that the assumptions below hold in the rest of this section.
\begin{asu} \label{assumption_case2}
\quad

\begin{enumerate}
\item The investor begins with initial portfolio $(x_0, 0)$, i.e., the investor does not hold any risky asset at the beginning, $y_0=0$.

\item The risky return $R$ in the market is modeled by \eqref{binomial_model}.

\item The utility function is a piece-wise exponential utility given by \eqref{exp_utility} with $\eta_+= \eta_- = \eta>0$.
\end{enumerate}

\end{asu}

\begin{rem}
The reference point $B$, given by \eqref{reference_point}, becomes
\[ B = (1+r)x_0 \;\;\text{when}  \;\; y_0=0.\]
Recall \eqref{replication_cost1} and \eqref{replication_cost2}, $B$ can be also rewritten as
\[ B = p^b_u \cdot \xi_u + p^b_d \cdot \xi_d = p^s_u \cdot \xi_u + p^s_d \cdot \xi_d .\]

Since $y_0=0$, we set $x_\xi = x_0$ and only consider investment strategies with initial wealth $x_0$.
Note that if $\theta_\xi$ is a replication strategy, given by \eqref{replication_strategy1} or \eqref{replication_strategy2}, $W(\theta_\xi) = \xi$ and $J(\theta_\xi) = V(\xi)$.

In Assumption \ref{assumption_case2}, we consider a piece-wise exponential utility function.
A major difference between a piece-wise power utility
and a piece-wise exponential utility  is that prospect utility under the exponential utility is always finite, due to the fact that $0 \le u_\pm(x) \le \zeta$ for all $x \ge 0$. Hence Assumption \ref{assumption_finiteness_CPT} is always satisfied under the exponential utility.
\end{rem}

\subsection{Main Results}
\label{subsec_2case_main}
We separate the prime problem $\sup_{\theta \in \R} J(\theta)$ into two sub problems:
\[ \text{(P3)} \; \sup_{\theta \ge 0} J(\theta)   \text{ and }  \text{(P4)}
\;  \sup_{\theta < 0} J(\theta) .\]

The results in the previous Remark motivate us to consider two sets of random payoffs:
\begin{align*}
\Xi^b:=\{ \xi=(\xi_u, \xi_d) \in \F_T: \xi_u \ge \xi_d, \; p^b_u \cdot(\xi_u - B) + p^b_d \cdot (\xi_d - B)=0 \}, \\
\Xi^s:=\{ \xi=(\xi_u, \xi_d) \in \F_T: \xi_u < \xi_d, \; p^s_u \cdot(\xi_u - B) + p^s_d \cdot (\xi_d - B)=0 \},
\end{align*}
and two sub problems:
\[ \text{(P3')} \; \sup_{\xi \in \Xi^b} V(\xi)   \text{ and }  \text{(P4')}
\;  \sup_{\xi \in \Xi^s} V(\xi) .\]
Notice that $\xi \in \Xi^b \Leftrightarrow \theta_\xi \ge 0$ and $\xi \in \Xi^s \Leftrightarrow \theta_\xi < 0$.\footnote{The results explain the superscript notations in $\Xi^b$ and $\Xi^s$ (``$b$" stands for ``buy" and ``$s$" stands for ``sell").}
In consequence, (P3) is equivalent to (P3') and (P4) is equivalent to (P4').
Once we have solved (P3') (or (P4')), we can easily obtain the solution to (P3) (or (P4)) by using \eqref{replication_strategy1} (or  \eqref{replication_strategy2}).
Finally, Problem \ref{main_problem} is solved by a comparison of the maximal value of (P3) and (P4).


To present the main theorems, we make the following definitions:
\begin{align}
  \bar{\zeta}_1 &:= \frac{w_+(1-p)}{ w_-(p)} &\,  \bar{\zeta}_2 &:=  \frac{p^b_d \cdot w_+(1-p)}{p^b_u \cdot w_-(p)}  , \label{def_bar_zeta}\\
  \underline{\zeta}_1 &:=  \frac{w_+(p)}{ w_-(1-p)}   &\,  \underline{\zeta}_2 &:= \frac{p^s_u \cdot w_+(p)}{p^s_d \cdot w_-(1-p)}, \label{def_underline_zeta}
\end{align}
and
\begin{align}
  \theta_3 &:= \frac{1}{\eta ( (1-\lambda)(u+d) - 2(1+r))} \ln \left( \frac{\ \bar{\zeta}_2 \ }{\zeta}\right), \label{def_theta3} \\
  \theta_4 &:= -\frac{1}{\eta \big( 2(1-\lambda)(1+r) - (u+d)   \big)} \ln \left( \frac{\ \underline{\zeta}_2 \ }{\ \zeta \ } \right). \label{def_theta4}
\end{align}

The solutions to (P3) and (P4) are summarized in Theorem \ref{thm_exp_case1} and Theorem \ref{thm_exp_case2}, respectively.
We provide detailed proofs for these two theorems in Subsections \ref{subsec_2case_p3} and \ref{subsec_2case_p4}.

\begin{thm}
\label{thm_exp_case1}
If Assumption \ref{assumption_case2} holds, we obtain explicit solutions to (P3) from the following cases.
\begin{enumerate}
  \item The optimal solution is $\theta^*=0$ and $J (\theta^*)=0$ if one of the following conditions holds:
  \vspace{1ex}
  \begin{enumerate}
    \item $p^b_d \le 0$;
    \vspace{1ex}
    \item $p^b_d = p^b_u > 0$ and $\zeta > \bar{\zeta}_1$;
    \vspace{1ex}
    \item $0< p^b_d < p^b_u $ and $\zeta \ge \bar{\zeta}_1$;
    \vspace{1ex}
    \item $p^b_d > p^b_u >0 $ and $\zeta \ge \bar{\zeta}_2$.
    \vspace{1ex}
  \end{enumerate}

  \item Any $\theta^* \in [0, +\infty)$ is an optimal solution and $J (\theta^*)=0$ if

  $p^b_d = p^b_u > 0$ and $\zeta = \bar{\zeta}_1$.
  \vspace{1ex}

  \item The optimal solution is $\theta^*= + \infty$ and $J (\theta^*)=  w_+(1-p) - \zeta \cdot w_-(p)>0$ if one of the following conditions holds:
  \vspace{1ex}
  \begin{enumerate}
    \item $p^b_d = p^b_u > 0$ and $\zeta < \bar{\zeta}_1$;
    \vspace{1ex}
    \item $0< p^b_d < p^b_u $ and $\zeta <\bar{\zeta}_1$.
    \vspace{1ex}
  \end{enumerate}

  \item The optimal solution is $\theta^* = \theta_3$ and $J (\theta^*)= J(\theta_3)>0$ if

  $p^b_d > p^b_u >0 $ and $\zeta < \bar{\zeta}_2$.
\end{enumerate}
\end{thm}

\begin{thm}
\label{thm_exp_case2}
If Assumption \ref{assumption_case2} holds, we obtain explicit solutions to (P4) from the following cases.
\begin{enumerate}
  \item The optimal solution is $\theta^*=0$ and $J (\theta^*)=0$ if one of the following conditions holds:
  \vspace{1ex}
  \begin{enumerate}
    \item $p^s_u \le 0$;
    \vspace{1ex}
    \item $p^s_u = p^s_d > 0$ and $\zeta > \underline{\zeta}_1$;
    \vspace{1ex}
    \item $0<p^s_u < p^s_d $ and $\zeta \ge \underline{\zeta}_1$;
    \vspace{1ex}
    \item $p^s_u > p^s_d >0 $ and $\zeta \ge \underline{\zeta}_2$.
    \vspace{1ex}
  \end{enumerate}

  \item Any $\theta^* \in (-\infty, 0)$ is an optimal solution and $J (\theta^*)=0$ if

  $p^s_u = p^s_d > 0$ and $\zeta = \underline{\zeta}_1$.
  \vspace{1ex}

  \item The optimal solution is $\theta^*= - \infty$ and $J (\theta^*)=  w_+(p) - \zeta \cdot w_-(1-p)>0$ if
  one of the following conditions holds:
  \vspace{1ex}
  \begin{enumerate}
    \item $p^s_u = p^s_d > 0$ and $\zeta < \underline{\zeta}_1$;
    \vspace{1ex}
    \item $0<p^s_u < p^s_d $ and $\zeta < \underline{\zeta}_1$.
    \vspace{1ex}
  \end{enumerate}

  \item The optimal solution is $\theta^* = \theta_4$ and $J (\theta^*)= J(\theta_4)>0$ if

  $p^s_u > p^s_d >0 $ and $\zeta < \underline{\zeta}_2$.
\end{enumerate}
\end{thm}

To simplify the citations from cases of the above two theorems, we introduce the notation [Th.4; 1] to denote Case (1) of Theorem \ref{thm_exp_case1}. Similar notations apply to Theorem \ref{thm_exp_case2} as well.

According to Theorem \ref{thm_exp_case1}, the maximal prospect of (P3) $J (\theta^*)$ takes three possible values:
$0$ in [Th.4; 1,2]; $w_+(1-p) - \zeta \cdot w_-(p)>0$ in [Th.4; 3]; and $J(\theta_3)>0$ in [Th.4; 4].
Regarding (P4), its maximal prospect $J (\theta^*)$ also takes three possible values:
$0$ in [Th.5; 1,2]; $w_+(p) - \zeta \cdot w_-(1-p)>0$ in [Th.5; 3]; and $J(\theta_4)>0$ in [Th.5; 4].
By comparing the two maximal prospect of (P3) and (P4), we find the maximal prospect and the optimal investment strategy to Problem \ref{main_problem}.
The results in Theorems \ref{thm_exp_case1} and \ref{thm_exp_case2} immediately lead to Theorem \ref{thm_exp} below.

\begin{thm}
\label{thm_exp}
If Assumption \ref{assumption_case2} holds, we obtain the optimal investment strategy $\theta^*$ to Problem \ref{main_problem} through the following cases.
\begin{enumerate}
\item $\theta^*=0$ if [Th.4; 1] and [Th.5; 1] hold simultaneously.
\vspace{1ex}

\item $\theta^* = \theta_3$ if [Th.4; 4] holds and one of the following satisfies as well:
\vspace{1ex}
\begin{enumerate}
  \item[(a)] [Th.5; 1, 2] holds;
  \vspace{1ex}

  \item[(b)] [Th.5; 3] holds and $J(\theta_3) \ge J(-\infty)$;
  \vspace{1ex}

  \item[(c)] [Th.5; 4] holds and $J(\theta_3) \ge J(\theta_4)$.
  \vspace{1ex}
\end{enumerate}

\item $\theta^* = \theta_4$ if [Th.5; 4] holds and one of the following satisfies as well:
\vspace{1ex}
\begin{enumerate}
  \item[(a)] [Th.4; 1, 2] holds;
  \vspace{1ex}

  \item[(b)] [Th.4; 3] holds and $J(\theta_4) \ge J(+\infty)$;
  \vspace{1ex}

  \item[(c)] [Th.4; 4] holds and $J(\theta_4) \ge J(\theta_3)$.
  \vspace{1ex}
\end{enumerate}

\item Any $\theta^* \in [0, +\infty)$ is an optimal investment strategy if [Th.4; 2] and [Th.5; 1] hold simultaneously.
\vspace{1ex}

\item Any $\theta^* \in (-\infty, 0]$ is an optimal investment strategy if [Th.4; 1] and [Th.5; 2] hold simultaneously.
\vspace{1ex}

\item Any $\theta^* \in (-\infty, +\infty)$ is an optimal investment strategy if [Th.4; 2] and [Th.5; 2] hold simultaneously.
\vspace{1ex}

\item $\theta^* = +\infty$ if [Th.4; 3] holds and one of the following satisfies as well:
\vspace{1ex}
\begin{enumerate}
  \item[(a)] [Th.5; 1, 2] holds;
  \vspace{1ex}

  \item[(b)] [Th.5; 3] holds and $J(+\infty) \ge J(-\infty)$;
  \vspace{1ex}

  \item[(c)] [Th.5; 4] holds and $J(+\infty) \ge J(\theta_4)$.
  \vspace{1ex}
\end{enumerate}

\item $\theta^* = -\infty$ if [Th.5; 3] holds and one of the following satisfies as well:
\vspace{1ex}
\begin{enumerate}
  \item[(a)] [Th.4; 1, 2] holds;
  \vspace{1ex}

  \item[(b)] [Th.4; 3] holds and $J(-\infty) \ge J(+\infty)$;
  \vspace{1ex}

  \item[(c)] [Th.4; 4] holds and $J(-\infty) \ge J(\theta_3)$.
\end{enumerate}
\end{enumerate}
\end{thm}

\begin{rem} \label{rem_case2_sensitivity}
Recall the definitions in \eqref{def_theta3} and \eqref{def_theta4}, two candidates for the optimal investment $\theta_3$ and $\theta_4$ decrease in absolute amount when the risk aversion parameter $\eta$ increases.

If $\lambda \ge \bar{\lambda}:=\max\{1 - \frac{1+r}{u}, \, 1 - \frac{d}{1+r}\}$, both $p^b_d, \, p^s_u \le 0$, then by  Theorem \ref{thm_exp}, the optimal investment $\theta^* = 0$.  This result shows the optimal investment largely depends on transaction costs. CPT investors will not trade the risky asset as long as $\lambda$ is above the threshold $\bar{\lambda}$.
However, if there are no transaction costs in the market ($\lambda = 0$), then the non-arbitrage condition $d < 1+r < u$ implies that $\bar{\lambda} > \lambda =0$.
\end{rem}

We close this subsection by discussing the results of the above theorems when $\lambda=0$.
If $\lambda = 0$, the financial market is a frictionless one.
Recall the definitions of $p^b_u$, $p^b_d$, $p^s_u$, and $p^s_d$. We simplify them as
\[  p^b_u = p^s_u = \frac{1+r-d}{u-d} := p_u>0, \;\; \text{and} \;\; p^b_d = p^s_d = \frac{u-(1+r)}{u-d} := p_d>0 ,\]
where we have used the non-arbitrage condition to derive $p_u, \, p_d>0$.

In consequence, we can drop all the cases with $p^b_d \le 0$ or $p^s_u \le 0$ in Theorems \ref{thm_exp_case1}, \ref{thm_exp_case2}, and \ref{thm_exp}.
In addition, it is easier to list the results under three scenarios: $p_u = p_d$; $p_u>p_d$; and $p_u<p_d$, which shall reduce the cases of the theorems above.
For instance, if $p_u<p_d$, then the optimal solution to Problem \ref{main_problem} is given by
\begin{equation*}
  \theta^* = \begin{cases}
  0, & \text{ if } \zeta \ge \max\{\bar{\zeta}_2, \, \underline{\zeta}_2\} \\
  \theta_3, & \text{ if } \underline{\zeta}_2 \le \zeta < \bar{\zeta}_2 \text{ or } \zeta < \min\{\bar{\zeta}_2, \, \underline{\zeta}_2\} \text{ and } J(\theta_3) \ge J(\theta_4)
\\
\theta_4, & \text{ if } \underline{\zeta}_2 > \zeta \ge \bar{\zeta}_2 \text{ or } \zeta < \min\{\bar{\zeta}_2, \, \underline{\zeta}_2\} \text{ and } J(\theta_3) \le J(\theta_4)  \end{cases}.
\end{equation*}
The above presentation of $\theta^*$ is also used in Subsections \ref{subsec_2case_p3} and \ref{subsec_2case_p4}.
Reformulating Theorems \ref{thm_exp_case1}, \ref{thm_exp_case2}, and \ref{thm_exp} using the above presentation method for $\lambda=0$ is straightforward, and is left as an exercise to interested readers.

\subsection{Solution to Sub-Problem (P3)}
\label{subsec_2case_p3}
In this subsection, we provide analysis and proofs for Theorem \ref{thm_exp_case1}.
Since (P3) and (P3') are equivalent, we focus on (P3'): $\sup_{\xi \in \Xi^b} V(\xi)$.

If $p^b_d \le 0$, i.e., $u \le \frac{1+r}{1-\lambda}$ (corresponding to $\Po(A_1)=1$ in Subsection \ref{subsection_1case_p1}), then $V(\xi) \le 0= V((B,B))$.
Hence $\xi^* = (B,B) \in \Xi^b$, and $\theta^*=0$ because of \eqref{replication_strategy1}.

In the rest of this subsection, we study the non-trivial case where $p^b_d \in (0,1)$.
Immediately, we have $\xi_d - B \le 0 \le \xi_u - B$, and
$\forall \, \xi \in \Xi^b$,
\[B - \xi_d = \frac{p^b_u}{p^b_d} (\xi_u - B).\]

By the definition of CPT, we write $V(\xi)$ as
\begin{align*}
V(\xi) &= w_+(1-p)\cdot u_+(\xi_u - B) - w_-(p) \cdot u_-(B -\xi_d) \\
&= w_+(1-p)\cdot u_+(\xi_u - B) -  w_-(p) \cdot u_- \left( \frac{p^b_u}{p^b_d} (\xi_u - B) \right) := L^b(\xi_u).
\end{align*}
Then sub-problem (P3') is equivalent to $ \sup_{\xi_u \ge B} L^b(\xi_u)$.

\begin{itemize}
\item $p^b_u = p^b_d$

In this case, using Assumption \ref{assumption_case2}, we rewrite $L^b(\xi_u)$ as
\[L^b(\xi_u) = \big( w_+(1-p) - \zeta \cdot w_-(p) \big) \cdot u_+(\xi_u - B).\]

Since $u_+(\xi_u - B)$ is an increasing function of $\xi_u$ and $w_+(1-p)>0$,
\[ \text{sign}((L^b)'(\xi_u)) = \; \text{sign} (\bar{\zeta}_1 - \zeta).\]
Recall $\bar{\zeta}_1 = w_+(1-p)/w_-(p)$ defined in \eqref{def_bar_zeta}.

Therefore, we obtain the optimal payoff $\xi_u^*$  by
\begin{equation*}
\xi_u^*=\begin{cases}
B, &\text{ when } \zeta > \bar{\zeta}_1 \\
[B, +\infty), &\text{ when } \zeta =\bar{\zeta}_1 \\
+\infty, &\text{ when } \zeta < \bar{\zeta}_1
\end{cases}.
\end{equation*}

Hence, using \eqref{replication_strategy1} and $\xi_d=2B-\xi_u$, the optimal solution $\theta^*$ to (P3) in $[0, +\infty)$ is given by
\begin{equation*}
\theta^*=\begin{cases}
0, &\text{ when } \zeta > \bar{\zeta}_1 \\
[0,+\infty), &\text{ when } \zeta =\bar{\zeta}_1 \\
+\infty, &\text{ when } \zeta < \bar{\zeta}_1
\end{cases}.
\end{equation*}

\item $p^b_u > p^b_d$

We calculate $(L^b)'(\xi_u)$ as
\begin{equation*}
(L^b)'(\xi_u) =  w_+(1-p) \cdot \eta e^{-\eta (\xi_u-B)} \left(1- \frac{\zeta}{\bar{\zeta}_2} \, e^{-\eta (p^b_u/p^b_d -1) (\xi_u-B)} \right).
\end{equation*}

If $\zeta \le \bar{\zeta}_2$, then $(L^b)'(\xi_u) > 0$ for all $\xi_u > B$.
The prospect $L^b(\xi_u)$ is a
strictly increasing function of $\xi_u$ (and thus $\theta$), hence the optimal investment in $[0, +\infty)$ is $\theta^*=+\infty$.

If $\zeta > \bar{\zeta}_2$, we obtain
\begin{align*}
(L^b)'(\xi^*_u) = 0 & \Leftrightarrow \xi^*_u = B +\frac{\ln \left( \zeta/\bar{\zeta}\right) }{ \eta \left( \frac{p^b_u}{p^b_d} -1 \right)}>B, \\
(L^b)'(\xi_u) \gtrless  0 & \Leftrightarrow \xi_u \gtrless \xi^*_u, \text{i.e., }
\xi^*_u \text{ is minimum}.
\end{align*}

The constraint $\theta \in [0, +\infty)$ is equivalent to $\xi_u \in [B, +\infty)$.
The maximizer of $\sup_{\xi_u} L^b(\xi_u)$ is $B$ or $+\infty$, namely,
\[ \sup_{\xi_u \in [B, +\infty)} L^b(\xi_u) = \max \{L^b(B)=0, \, L^b(+\infty)\},\]
where
\[ L^b(+\infty) := \lim_{\xi_u \to +\infty} L^b(\xi_u) = w_+(1-p) - \zeta \cdot w_-(p).\]
Apparently, $L^b(+\infty) > 0 \Leftrightarrow \zeta < \bar{\zeta}_1$.

In this scenario,
\[ \bar{\zeta}_2 = \frac{p^b_d}{p^b_u} \cdot \bar{\zeta}_1 < \bar{\zeta}_1.\]
Then we obtain the optimal solution $\theta^*$ to (P3) in $[0, +\infty)$
\begin{equation*}
\theta^*= \begin{cases}
0, &\text{ if } \zeta \ge \bar{\zeta}_1 \\
+\infty, & \text{ if } \zeta < \bar{\zeta}_1
\end{cases}.
\end{equation*}

\item $p^b_u < p^b_d$

Due to the analysis above, we easily obtain that $(L^b)'(\xi_u) < 0$ when $\zeta \ge \bar{\zeta}_2$, and thus $\theta^* = 0$ in this scenario.

If $\zeta < \bar{\zeta}_2$, solving $(L^b)'(\xi_u) =0$ gives
\[\xi_u^* = B + \frac{p^b_d}{\eta (p^b_d - p^b_u)} \ln \left( \frac{\ \bar{\zeta} \ }{\zeta}\right)>B,\]
and the corresponding replication strategy is obtained using \eqref{replication_strategy1}
\[ \frac{\xi^*_u - B}{p^b_d(1-\lambda)(u-d)}  =
\frac{\ln \left( \bar{\zeta}/\zeta \right)}{\eta (1-\lambda)(u-d)(p^b_d - p^b_u)} =\theta_3,\]
where $\theta_3$ is defined by \eqref{def_theta3}.

Since $(L^b)'(\xi_u) \gtrless  0  \Leftrightarrow \xi_u \lessgtr \xi^*_u$, $\xi^*_u$ is the unique maximizer to the problem $\sup_{\xi_u \ge B} L^b(\xi_u)$.

Notice that $\theta_3 > 0$ due to $p^b_u < p^b_d $ and  $\zeta < \bar{\zeta}_2$.

Therefore, if $p^b_u < p^b_d$, the optimal solution $\theta^*$ to (P3) in $[0, +\infty)$ is
\begin{equation*}
\theta^*= \begin{cases}
0, &\text{ if } \zeta \ge \bar{\zeta}_2 \\
\theta_3, & \text{ if } \zeta < \bar{\zeta}_2
\end{cases}.
\end{equation*}

\end{itemize}

\subsection{Solution to Sub-Problem (P4)}
\label{subsec_2case_p4}
To solve (P4) and show Theorem \ref{thm_exp_case2}, we study (P4') in this subsection.

If $p^s_u \le 0$, then $V(\xi) \le 0$ for all $\xi \in \Xi^s$, and $\theta^*=0$.

In the remaining part, we consider the non-trivial case where $p^s_u >0$.
$\forall \, \xi \in \Xi^s$, we have
$ \xi_u - B \le 0 \le \xi_d - B$, and
\[B - \xi_u = \frac{p^s_d}{p^s_d} (\xi_d -B).\]
Hence,
\begin{align*}
V(\xi) &= w_+(p) \cdot u_+(\xi_d - B) - w_-(1-p) \cdot  u_-(B-\xi_u) \\
&=w_+(p)\cdot  u_+(\xi_d - B) - w_-(1-p) \cdot u_-\left( \frac{p^s_d}{p^s_u} (\xi_d - B) \right)
:=L^s(\xi_d).
\end{align*}
The first derivative of $L^s(\xi_d)$ is calculated as
\[ (L^s)'(\xi_d) = \eta e^{-\eta (\xi_d - B)} w_+(p) \left[ 1 - \frac{\zeta}{\underline{\zeta}_2}
e^{-\eta \left(\frac{p^s_d}{p^s_u} -1 \right) (\xi_d -B)} \right],\]
where constant $\underline{\zeta}_2$ is defined by \eqref{def_underline_zeta}.

By \eqref{replication_strategy2}, we derive
\[ \theta = - \frac{\xi_d - B}{p^s_u(u-d)}.\]

It is obvious that sub-problem (P4') and $\sup_{\xi_d \ge B} L^s(\xi_d)$ are equivalent.
The analysis is the same as for $\sup_{\xi_u \ge B} L^b(\xi_u)$ in the previous subsection, and we summarize results below.

\begin{itemize}
\item $p^s_u = p^s_d$

In this case, we have
\[ \text{sign} \big( (L^s)'(\xi_d) \big) =  \text{ sign} \left( 1 - \frac{\zeta}{\underline{\zeta}_1} \right)=-\text{sign} \left( J'(\theta)\right) , \]
and then
\begin{equation*} 
\theta^*=\begin{cases}
0, &\text{ when } \zeta > \underline{\zeta}_1 \\
(-\infty, 0), &\text{ when } \zeta =\underline{\zeta}_1 \\
-\infty, &\text{ when }  \zeta < \underline{\zeta}_1
\end{cases}.
\end{equation*}

\item $p^s_u < p^s_d$

If $\zeta \le \underline{\zeta}_2$, then $ (L^s)'(\xi_d) >0$ for all $\xi_d > B$ and
$J'(\theta) <0$ for all $\theta<0$. The optimal solution is thus $\theta^*= - \infty$.

If $\zeta > \underline{\zeta}_2$, we obtain
\begin{align*}
(L^s)'(\xi^*_d) = 0 & \Leftrightarrow \xi^*_d =  B + \frac{p^s_u}{\eta (p^s_u - p^s_d)} \ln \left( \frac{\ \underline{\zeta}_2 \ }{\ \zeta \ } \right)>B, \\
(L^s)'(\xi_d) \gtrless  0 & \Leftrightarrow \xi_d \gtrless \xi^*_d, \text{i.e., }
\xi^*_d \text{ is minimum}.
\end{align*}

Therefore, the maximum will be achieved at the end or limit point, i.e.,
\[ \sup_{\theta \in (-\infty, \, 0]} J(\theta) = \max\{ J(0)=0,  \, J(-\infty)\},\]
where
\[J(-\infty) := \lim_{\theta \to -\infty} J(\theta) = \lim_{\xi_d \to +\infty} L^s(\xi_d) = w_+(p) - \zeta \cdot w_-(1-p) .\]

Notice that $J(-\infty) > 0 \Leftrightarrow \zeta < \underline{\zeta}_1$ and
\[ \underline{\zeta}_2 = \frac{p^s_u}{p^s_d} \cdot \underline{\zeta}_1 < \underline{\zeta}_1.\]

In conclusion, if $p^s_u < p^s_d$,
the optimal solution $\theta^*$ to (P4) in $(-\infty, 0]$ is given by
\begin{equation*} 
\theta^* = \begin{cases}
0, &\text{ if } \zeta \ge \underline{\zeta}_1  \\
-\infty, & \text{ if } \zeta < \underline{\zeta}_1
\end{cases}.
\end{equation*}

\item $p^s_u > p^s_d$

If $\zeta \ge \underline{\zeta}_2$, then $ (L^s)'(\xi_d) < 0$ for all $\xi_d > B$ and $J'(\theta)>0$ for all $\theta<0$, so the optimal solution $\theta^*=0$.

If $\zeta < \underline{\zeta}_2$, we obtain
\[ (L^s)'(\xi_d) \gtrless  0  \Leftrightarrow \xi_d \lessgtr \xi^*_d,\]
which implies that $\xi^*_d = \argmax L^s(\xi_d)$.

For $(\xi^*_u, \xi^*_d)$, where $\xi^*_u$ satisfies $p^s_u \cdot (\xi^*_u - B) + p^s_d \cdot (\xi^*_d - B)=0$, the corresponding replication strategy is given by
\[\frac{\xi^*_u - \xi^*_d}{u - d} = - \frac{\xi^*_d - B}{p^s_u (u -d)}
= - \frac{\ln \left( \underline{\zeta}_2 / \zeta \right) }{\eta (u-d) (p_u^s - p_d^s)} =\theta_4, \]
where $\theta_4$ is defined by \eqref{def_theta4}.

We then obtain the optimal solution $\theta^*$ to (P4) in $(-\infty, 0]$ as
\begin{equation*} 
\theta^*= \begin{cases}
0, &\text{ if } \zeta \ge \underline{\zeta}_2 \\
\theta_4, & \text{ if } \zeta < \underline{\zeta}_2
\end{cases}.
\end{equation*}
\end{itemize}

\section{Economic Analysis}
\label{sec_economic_analysis}
In this section, we conduct an economic analysis to study how the optimal investment strategy is affected by transaction costs and risk aversion. The calculations in Section \ref{sec_case2} are straightforward as long as the binomial model \eqref{binomial_model} has been estimated. However, under Tversky and Kahneman's weighting functions \eqref{weighting_TK92} or Prelect's weighting functions \eqref{weighting_P98}, the numerical calculations for $K_1(Z_1)$ and $K_2(Z_2)$ (two integrals) in Section \ref{section_case1} are very complicated even when the risky return $1+R$ is normally distributed or lognomarlly distributed.
In what follows, we obtain numerical results based on the model in Section \ref{section_case1} and conclusions from Theorem \ref{thm_1case}.

\subsection{Data and Model Parameters}
\label{subsec_data}
We consider optimal investment problems in a single-period discrete model, so we select a relatively short time window.
In the economic analysis thereafter, we select the time window to be 1 week, $T=1$ week.

To estimate the risk-free interest rate $r$, we use 3-month EONIA  (Euro OverNight Index Average) Swap Index bid close quotes
between January 2, 2012 and June 30, 2015.
There are 891 daily observations during the selected time period.\footnote{Data source: Thomson Reuters Eikon. Access from the Chair of Mathematical Finance at the Technical University of Munich is greatly appreciated.}
To have more consistent data, we convert the daily frequency into weekly.
The descriptive statistics for the weekly quotes are summarized in the table below.
\vspace{-1em}
\begin{table}[!h] \label{table_summary}
\renewcommand{\arraystretch}{1.5}
\begin{tabular}{|c|c|c|c|c|} \hline
Obs. & Mean & Median & Std. & Skewness \\ \hline
183 & 0.0779\% & 0.0696\% & 0.1240 & 0.4317 \\ \hline
\end{tabular}
\caption{Summary Statistics of Annualized Risk-free Return}
\end{table}

\vspace{-2em}
Due to the right skewness, we choose the median as the estimate for the risk-free return.
Then the weekly risk-free return $r$ is obtained by
\[ r = (1 + 0.0696\%)^{1/52} - 1 = 1.3380 \times 10^{-5}.\]

In order to estimate the distribution of the risky return $R$, we choose the weekly close quotes of FTSE (Financial Times Stock Exchange) 100 Index from January 2, 2012 to July 6, 2015.\footnote{Data source: Yahoo Finance https://uk.finance.yahoo.com/q/hp?s=\%5EFTSE.}
We calculate the log return of the FTSE 100 index and obtain
\[ \mu=3.2932 \times 10^{-4}, \text{ and } \sigma=7.4383 \times 10^{-3}.\]
\begin{figure}
  \centering
  \includegraphics[width=\textwidth]{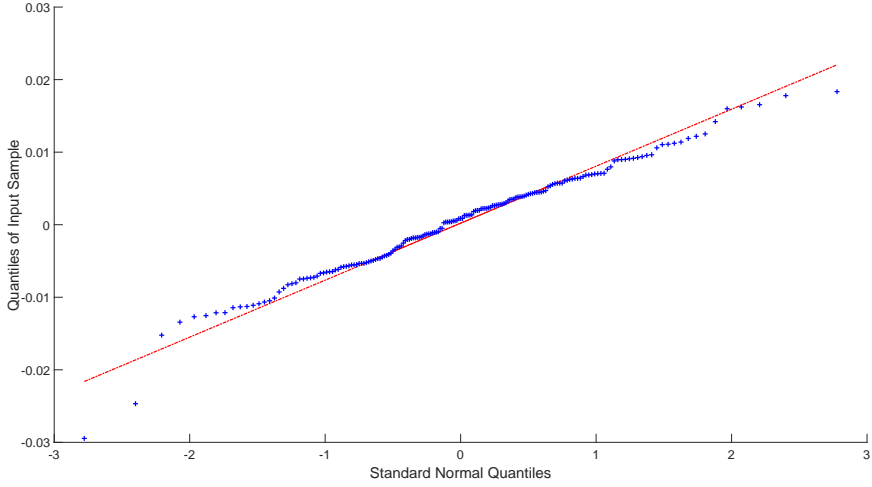}\\
  \caption{QQ Plot of $\ln(1+R)$ versus Standard Normal}
  \label{fig_qq_weekly}
\end{figure}
The QQ plot of $\ln(1+R)$ versus standard normal in Figure \ref{fig_qq_weekly} suggests that $\ln(1+R)$ is approximately normal. From now on, we assume
\[\ln(1+R) \sim N\left(\mu, \; \sigma^2 \right).\]

For the numerical calculations in this section, we select Tversky and Kahneman's weighting function (given by \eqref{weighting_TK92}) with parameters
\[\gamma = 0.61 \text{ and } \delta = 0.69.\]

We consider a piece-wise power utility function given by \eqref{utility_TK92}. The risk attitudes of an CPT investor depend on $\alpha$ and $\beta$. We separate the discussions into two cases: $\alpha=\beta$ and $\alpha<\beta$.

\subsection{The Case of $\alpha = \beta$}
If $\alpha = \beta$, the optimal investment strategy is given by the corresponding cases in Theorem \ref{thm_1case}.
In the analysis, we select
$\alpha = \beta = 0.88 $, as estimated in \cite{TK92}.

Since $\ln(1+R)$ is normally distributed, we have $ 0 < \Po(A_1), \, \Po(A_2) <1$.
Then according to Theorem \ref{thm_1case}, we need to calculate $K_1(Z_1)$ and $K_2(Z_2)$ in order to obtain the optimal investment strategy $\theta^*$.
The graphs of $K_1(Z_1)$ and $K_2(Z_2)$ as a function of transaction cost parameter $\lambda$  are provided in Figure \ref{fig_k1_weekly} and Figure \ref{fig_k2_weekly}, respectively.

\begin{figure}[!htb]
\vspace{-1em}
  \centering
  \includegraphics[width=\textwidth]{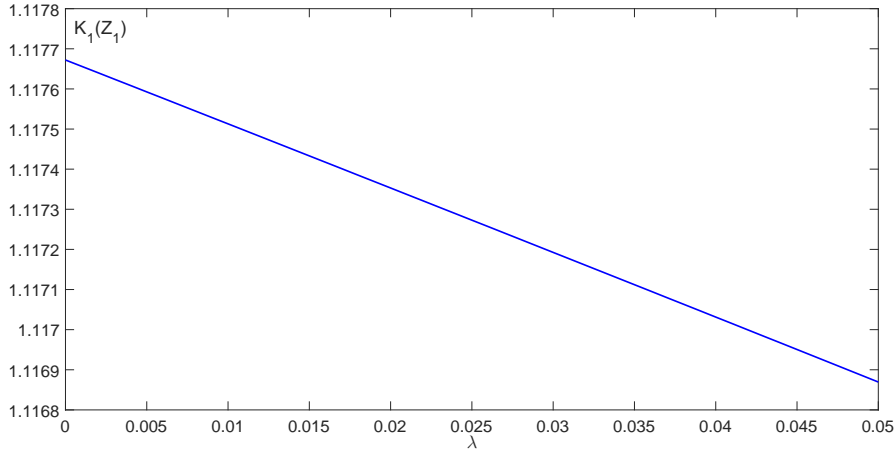}\\
  \caption{$K_1(Z_1)$ when $0<\lambda <5\%$}
  \label{fig_k1_weekly}
  \vspace{-1em}
\end{figure}

\begin{figure}[!htb]
  \centering
  \includegraphics[width=\textwidth]{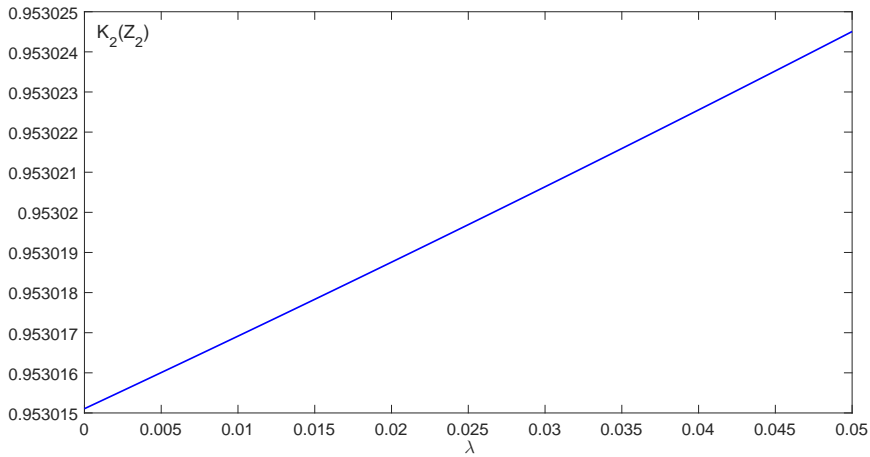}\\
  \caption{$K_2(Z_2)$ when $0<\lambda <5\%$}
  \label{fig_k2_weekly}
  \vspace{-1em}
\end{figure}

If $\lambda$ increases, i.e., $\lambda \uparrow$ (recall $Z_1 = (1-\lambda)(1+R) - (1+r)$ and $Z_2 = (1-\lambda) (R-r)$),
both $Z_1$ and $Z_2$ will decrease ($Z_1 \downarrow$ and $Z_2 \downarrow$). Then immediately, we obtain $F_{Z_1} \uparrow$, $S_{Z_1} \downarrow$, $F_{Z_2} \uparrow$, and $S_{Z_2} \downarrow$, which, by definitions \eqref{def_g1_l1} and \eqref{def_g2_l2}, imply that $g_1(Z_1) \downarrow$, $l_1(Z_1) \uparrow$, $g_2(Z_2) \uparrow$,
and $l_2(Z_2) \downarrow$. All these results together suggest that $K_1(Z_1) \downarrow$ and $K_2(Z_2) \uparrow$, which are confirmed by Figure \ref{fig_k1_weekly} and Figure \ref{fig_k2_weekly}.

From Figures \ref{fig_k1_weekly} and \ref{fig_k2_weekly}, we observe that $1 < K_1(Z_1) < 2.25$ and $K_2(Z_2) < 1 $ for all $\lambda \in (0, 5\%)$. In \cite{TK92}, $k$ is estimated to be 2.25, then Case (1d) in Theorem \ref{thm_1case} holds, and hence we obtain the optimal investment $\theta^*=0$.

In this numerical example, the time window is chosen as one week and we have a ``bear" market after the financial crisis of 2007-2008 during the selected period; thus, the difference between investment returns $R(\omega)-r$ is small for most states $\omega \in \Omega$.  With a longer time window and/or a better market performance, $R-r$ will increase, resulting in the increase of $Z_1$ and $Z_2$. Hence, we infer $K_1(Z_1)$ will be greater than 2.25 at certain model/market conditions when transaction costs are small.
On the other hand, despite $K_2(Z_2)$ is an increasing function of $\lambda$ (then a decreasing function of $R-r$), $K_2(Z_2)$ is less sensitive to the change of $\lambda$ or $R-r$ comparing to $K_1(Z_1)$.
Therefore, in a ``bull" market, we may have the case
$\max\{K_1(Z_1), \, K_2(Z_2)\} = K_1(Z_1)  > k$ for small $\lambda$, which corresponds to Case (8b) in Theorem \ref{thm_1case}, and then $\theta^* = +\infty$. The economic interpretation for this scenario is that CPT investors should buy the risky asset as much as they can in a very good economy.
For example, if we assume the price process of the risky asset is given by a geometric Brownian Motion with drift $15\%$ and volatility $20\%$
and the risk-free interest rate is $r=5\%$. In addition, we select $\lambda = 1\%$ and $T = 1 $ year.
We find $K_1(Z_1) = 2.7144 > k = 2.25$ and $K_2(Z_2) = 0.3957$, implying $\theta^* = +\infty$ in this market model.
Clearly, if $\lambda$ is large enough (e.g., $\lambda \to 1$), the optimal investment will be $0$.
In scenarios when $K_1(Z_1)  > k$ for small $\lambda$, the impact of transaction costs on the optimal investment $\theta^*$ is dramatic, because $\theta^*=+\infty$ if $\lambda$ is less than a critical threshold, but $\theta^*=0$ if $\lambda$ is greater than the threshold.

\subsection{The Case of $\alpha < \beta$}
We next study the case of $\alpha < \beta$ when $\ln(1+R)$ is normally distributed.
We investigate the impact of the utility parameters, $\alpha$ and $\beta$, on the optimal investment strategy.
In this particular study, we assume the investment constraint is not binding, and hence,
\[\theta^*=\argmax_{\{\theta_1, \, \theta_2\}} J(\theta),\]
where $\theta_1$ and $\theta_2$ are given by \eqref{def_theta1_theta2}.
Remark \ref{rem_case1} of Theorem \ref{thm_1case} provides conditions when $\theta^*=\theta_1 \text{ or } \theta_2$, see also \cite[Appendix]{HZ11}.

  \begin{figure}[!htb]
  \vspace{-2em}
  \centering
  \includegraphics[width=\textwidth]{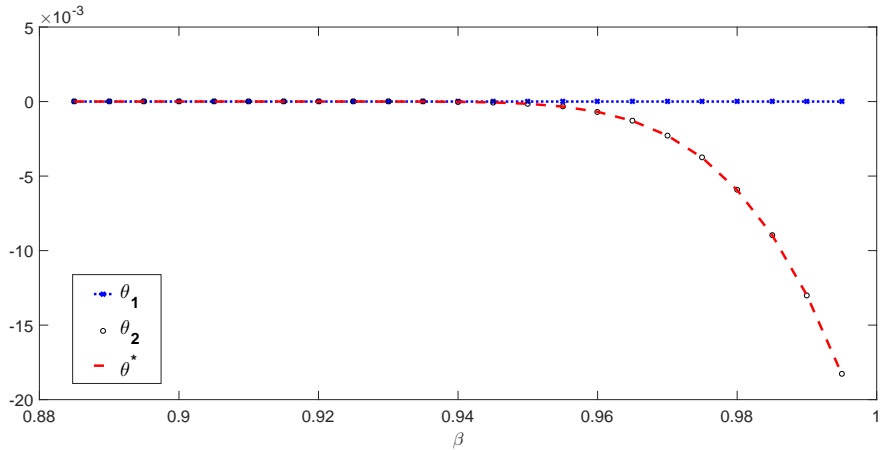}\\
  \caption{$\alpha=0.88$, $0.88 < \beta < 1$, and $\lambda=1\%$}
  \label{fig_theta_beta}
   \vspace{-1em}
\end{figure}

First, we fix $\alpha = 0.88$, and calculate $\theta_1$ and $\theta_2$ as functions of $\beta$, where $ 0.88 < \beta < 1$.
The transaction cost parameter $\lambda$ is chosen at $1\%$.
In Figure \ref{fig_theta_beta}, the line marked in circle coincides with the dashed line, i.e., $\theta^*=\theta_2$, implying that the optimal strategy is to short the risky asset.
Furthermore, we observe that $\theta_1$ is an increasing function of $\beta$,\footnote{The increasing property of $\theta_1$ with respect to $\beta$ is not that noticeable in Figure \ref{fig_theta_beta}, but is clearly supported by numerical values.} but $\theta_2$ is a decreasing function of $\beta$.
Therefore, the optimal investment (in absolute amount) increases when $\beta$ increases.

Next, we fix $\beta = 0.88$ and $\lambda = 1\%$, and consider $\alpha \in (0.6,0.88)$.
By following similar numerical calculations as in the previous study, we draw the graphs in Figure \ref{fig_theta_alpha}. Comparing with the findings from Figure \ref{fig_theta_beta}, we obtain exactly opposite results regarding monotonicity.
Namely, $\theta_1$ is a decreasing function of $\alpha$ and $\theta_2$ is an increasing  function of $\alpha$.
As before, we still have $\theta^*=\theta_2$.
Therefore, this study shows that the optimal investment (in absolute amount) decreases as $\alpha$ increases.

  \begin{figure}[!htb]
  \vspace{-2em}
  \centering
  \includegraphics[width=\textwidth]{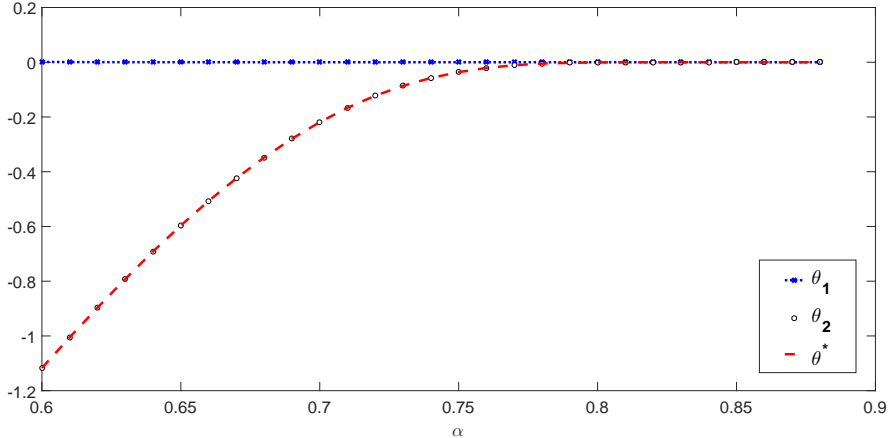}\\
  \caption{$0.6< \alpha<0.88$, $\beta = 0.88$, and $\lambda=1\%$}
  \label{fig_theta_alpha}
  \vspace{-1em}
\end{figure}

Lastly, we fix $\alpha = 0.8$ and $\beta = 0.88$, and consider $\lambda \in (0, 0.15\%]$ (between 0 and 15 bps).
The results in this case are drawn in Figure \ref{fig_theta_lambda}.
Notice that we have $\theta^* = \theta_1$ only when transaction costs are small  and $\theta^* = \theta_2$ otherwise.
This result shows that transaction costs are crucial to the optimal investment strategy.
Once the transaction cost parameter $\lambda$ increases beyond a certain threshold ($7 \sim 8$ bps in the numerical example), the optimal investment strategy will shift from ``long position" to ``short position" in the risky asset.
  \begin{figure}[!htb]
  \centering
  \includegraphics[width=\textwidth]{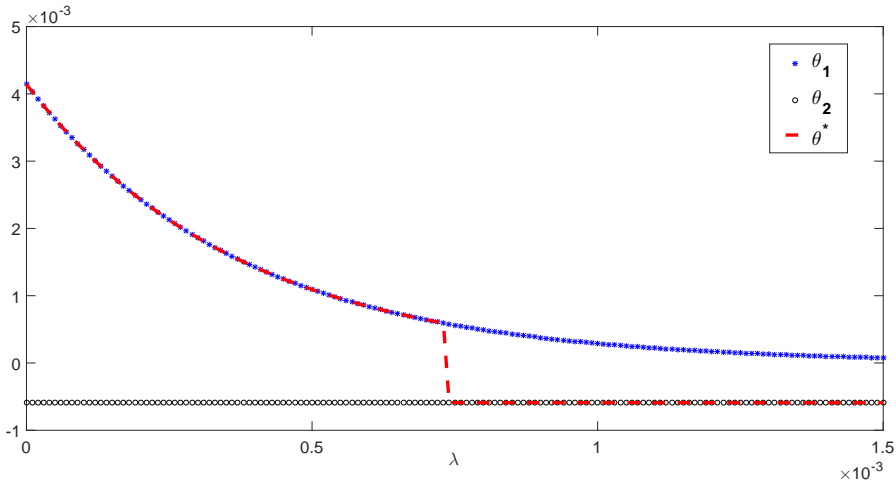}\\
  \caption{$\alpha=0.8$, $\beta = 0.88$, and $0< \lambda \le 0.15\%$}
  \label{fig_theta_lambda}
\end{figure}

\section{Conclusions}
\label{sec_conclusion}
Prospect theory was proposed in \cite{KT79}, and further developed into cumulative prospect theory (CPT) in \cite{TK92}. According to CPT, people evaluate uncertain outcomes by comparing them to a reference point, which separates all the outcomes into gains and losses based on the comparison. In addition, people's risk attitudes towards gains and losses are not universally risk averse. Instead, they exhibit fourfold patterns (see \cite{TK92}):
\begin{quote}
\emph{risk aversion for gains and risk seeking for losses of high probability;
risk seeking for gains and risk aversion for losses of low probability.}
\end{quote}
The experimental studies challenge some fundamental axioms of expected utility theory (EUT), which, by far, is still the most popular criterion in economics and finance when it comes to decision making with uncertainty.

In this paper, we consider a CPT investor in a single-period discrete-time financial model with transaction costs.
The investor seeks the optimal investment strategy that maximizes the prospect value of his/her final wealth.

The main objective of our work is to obtain explicit solutions to the optimal investment problem with transaction costs under CPT.
We obtain the optimal investment in explicit form to this problem in two examples when the utility function and the reference point are given in specific forms.
We conduct an economic analysis to study the impact of transaction costs and risk aversion on the optimal investment.
The results confirm that transaction costs play an important role in the optimal investment.
There exist thresholds for the transaction cost parameter $\lambda$.
In some cases, the optimal investment is 0 when $\lambda$ is above a threshold.
In other cases, there exists a threshold for $\lambda$ which separates the optimal investment into ``buy" strategies and ``sell" strategies.
When the investor's preference is characterized by a piece-wise power utility $u(x)=x^\alpha \cdot \bm{1}_{x \ge 0} - k (-x)^\beta \cdot \bm{1}_{x <0}$, we also observe that the optimal investment
decreases in amount as $\alpha$ increases or $\beta$ decreases.
However, such finding may not hold in general, as pointed out in \cite[Section 5]{BG10}.

\begin{acknowledgements}
The first author B.Z. acknowledges the financial support from the Technical University of Munich (TUM) through the TUM foundation fellowship.
This work was done when B.Z. was a postdoctoral fellow at the Chair of Mathematical Finance of TUM.
We are thankful to two anonymous referees for valuable comments and suggestions.
\end{acknowledgements}



\end{document}